\documentclass[journal]{IEEEtran}

\usepackage{amsmath,amsfonts}
\usepackage{algorithm}
\usepackage{algpseudocode}
\usepackage{array}
\usepackage[caption=false,font=normalsize,labelfont=sf,textfont=sf]{subfig}
\usepackage{textcomp}
\usepackage{stfloats}
\usepackage{url}
\usepackage{verbatim}
\usepackage{graphicx}
\usepackage{cite}
\usepackage[utf8]{inputenc}

\usepackage{mathtools}
\usepackage{bm}
\usepackage{amssymb}
\usepackage{array}
\usepackage{multirow, multicol}
\usepackage[pdfencoding=auto, psdextra]{hyperref}
\usepackage{cleveref}
\usepackage{amsthm} %
\usepackage{xcolor} %
\usepackage{crossreftools}  %
\usepackage{enumitem} %
\usepackage{siunitx} %
\usepackage{placeins}
\usepackage{lipsum}
\usepackage{microtype}

\DeclarePairedDelimiter\parens{\lparen}{\rparen}
\DeclarePairedDelimiter\brackets{[}{]}

\DeclareMathOperator{\diag}{diag}
\DeclareMathOperator{\trace}{tr}
\DeclareMathOperator{\matrixRank}{rank}
\DeclareMathOperator{\expectedValue}{\mathbb{E}}
\DeclareMathOperator{\blkdiag}{blkdiag}

\newcommand{\diagp}[1]{\diag\parens*{#1}}
\newcommand{\E}[1]{\expectedValue\brackets*{#1}}
\newcommand{\rank}[1]{\matrixRank\parens*{#1}}

\newcommand{\kron}{\otimes}

\newcommand{\etal}{\textit{et al}.\,}
\newcommand{\ie}{i.e., }
\newcommand{\eg}{\textit{e}.\textit{g}.\,}

\renewcommand{\v}[1]{\bm{#1}}

\def\noisevar{\gamma}

\newcommand{\sbeforekron}{\bar{s}}

\long\def\hide#1{}

\NewDocumentCommand{\grad}{e{_^}}{%
  \mathop{}\!%
  \mathop{}\mspace{-1mu}%
  \nabla
  \IfValueT{#1}{_{\!#1}}%
  \IfValueT{#2}{^{#2}}%
  \mspace{-1mu}
}

\newcommand\scalemath[2]{\scalebox{#1}{\mbox{\ensuremath{\displaystyle #2}}}}
\newtheorem{theorem}{Theorem}
\newtheorem{lemma}{Lemma}

\usepackage{xpatch}
\makeatletter
\AtBeginDocument{\xpatchcmd{\@thm}{\thm@headpunct{.}}{\thm@headpunct{}}{}{}}
\makeatother

\let\ORGhypersetup\hypersetup
\protected\def\hypersetup{\ORGhypersetup}
\newenvironment{enumerateproof}
 {\begin{enumerate}
 [font=\upshape,label=(\alph*)]}
 {\end{enumerate}}

 \makeatletter
\def\th@plain{%
  \thm@notefont{}%
  \itshape %
}
\def\th@definition{%
  \thm@notefont{}%
  \normalfont %
}
\makeatother

\newcommand{\blue}[1]{{#1}}


\hypersetup{
    colorlinks,
    linkcolor={red!90!black},
    citecolor={blue!90!black},
    urlcolor={blue!90!black}
}

\def\figwidthnormal{0.44\linewidth}

\usepackage{tikz}
\newcommand\copyrighttext{%
  \footnotesize \textcopyright 2025 IEEE.  Personal use of this material is permitted.  Permission from IEEE must be obtained for all other uses, in any current or future media, including reprinting/republishing this material for advertising or promotional purposes, creating new collective works, for resale or redistribution to servers or lists, or reuse of any copyrighted component of this work in other works.
  DOI: \href{10.1109/TASLPRO.2025.3533371}{10.1109/TASLPRO.2025.3533371}}
\newcommand\copyrightnotice{%
\begin{tikzpicture}[remember picture,overlay]
\node[anchor=south,yshift=0pt] at (current page.south) {\fbox{\parbox{\dimexpr\textwidth-\fboxsep-\fboxrule\relax}{\copyrighttext}}};
\end{tikzpicture}%
}

\begin{document}

\title{Wideband Relative Transfer Function (RTF) Estimation Exploiting Frequency Correlations}

\author{Giovanni Bologni, Richard C. Hendriks and Richard Heusdens
\thanks{
}%
\thanks{
Manuscript received ???; revised ???; ?????.
This work was partly supported by the Dutch Research Council (NWO) and partly by the Signal Processing Systems Group, Delft University of Technology, Delft, The Netherlands.
The associate editor coordinating the review of this manuscript and approving it for publication was ??? (Corresponding author: Giovanni Bologni.) 
Giovanni Bologni and Richard C. Hendriks are with the Faculty of Electrical Engineering, Mathematics, and Computer Science, Delft University of Technology, The Netherlands (e-mail: G.Bologni@tudelft.nl; R.C.Hendriks@tudelft.nl).
Richard Heusdens is with the Faculty of Military Sciences at the Netherlands Defence Academy (NLDA) (e-mail: r.heusdens@mindef.nl).
}
}

\markboth{IEEE/ACM TRANSACTIONS ON AUDIO SPEECH AND LANGUAGE PROCESSING
}%
{}

\maketitle
\copyrightnotice

\begin{abstract}
This article focuses on estimating relative transfer functions (RTFs) for beamforming applications.
\blue{
Traditional methods often assume that spectra are uncorrelated, an assumption that is often violated in practical scenarios due to factors such as time-domain windowing or the non-stationary nature of signals, as observed in speech.
}
To overcome these limitations, we propose an RTF estimation technique that leverages spectral and spatial correlations through subspace analysis.
Additionally, we derive Cramér--Rao bounds (CRBs) for the RTF estimation task, providing theoretical insights into the achievable estimation accuracy.
These bounds reveal that channel estimation can be performed more accurately if the noise or the target signal exhibits spectral correlations.
Experiments with both real and synthetic data show that our technique outperforms the narrowband maximum-likelihood estimator, known as covariance whitening (CW), when the target exhibits spectral correlations.
Although the proposed algorithm generally achieves accuracy close to the theoretical bound, there is potential for further improvement, especially in scenarios with highly spectrally correlated noise.
\blue{While channel estimation has various applications, we demonstrate the method using a minimum variance distortionless (MVDR) beamformer for multichannel speech enhancement. A free Python implementation is also provided.}
\end{abstract}


\vspace{-1mm}
\section{Introduction}\label{sec:intro}
\IEEEPARstart{S}{patial} filtering techniques can extract a target signal from the measurements of multiple sensors, also referred to as \textit{beamforming} \cite{gannot_consolidated_2017, hadad_binaural_2016}.
Most beamforming techniques, such as the minimum variance distortionless beamformer (MVDR) 
, 
rely on the knowledge of the relative transfer function (RTF) between a target emitter and a sensor array to virtually \textit{steer} the array towards the direction of interest \cite{gannot_signal_2001, doclo_acoustic_2010}.
RTFs generalize the angle or direction-of-arrival (DOA) concept in scenarios involving the proximity of the source to the receivers or the presence of reflections.
These scenarios commonly arise in acoustics and wireless communications, radar and sonar sensing, seismology, and medical imaging.

One fundamental assumption shared among many channel estimation techniques is that RTFs can be estimated independently per frequency bin after transforming the received signal to the short-time Fourier transform (STFT) domain \cite{cohen_relative_2004,koutrouvelis_robust_2019, hoang_joint_2021, li_low_2022, li_noise_2023, li_joint_2023}.
This implies that the signals are realizations of wide-sense stationary (WSS) processes or that distinct frequency components of the signal are mutually uncorrelated.
It was shown that distinct frequency components of a random process are statistically uncorrelated if and only if the process is WSS \cite{napolitano_-_2020}.
\begin{figure}[tb]
    \vspace{-3mm}
    \centering
    \setlength{\abovecaptionskip}{2pt plus 3pt minus 2pt} 
    \setlength{\belowcaptionskip}{4pt plus 3pt minus 2pt} 
    
    \includegraphics[width=\linewidth]{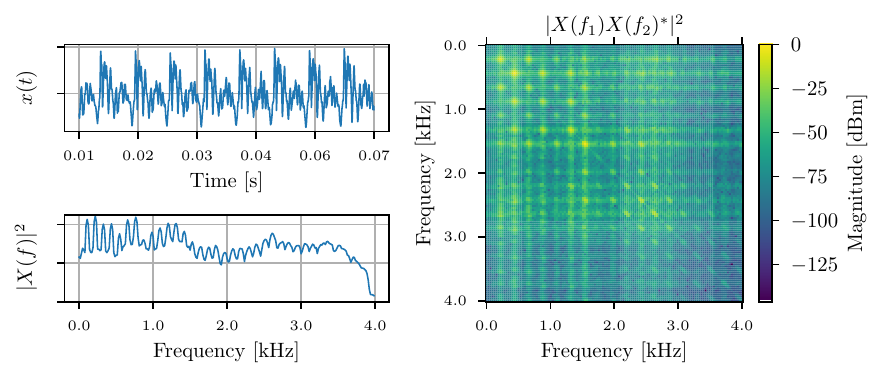}
    \caption{
    The /ä/ phoneme uttered by a male speaker. The top left plot depicts the waveform, while the bottom left plot shows the power spectral density (PSD). The peaks in the PSD are found at integer multiples of the fundamental frequency (harmonics).
    The right plot shows the spectral correlation or bifrequency spectrum. The grid-like structure of peaks in the bifrequency spectrum, whose spacing is proportional to the fundamental frequency, indicates a correlation between harmonic components \cite{madisetti_cyclostationary_2009}.
    }
    \label{fig:waveform}
\end{figure}

However, the spectral uncorrelation assumption is frequently violated in practice.
The STFT coefficients of the signals in neighboring frequency bands are correlated due to the use of short frame lengths and overlap-add/save techniques.
In wireless communications, non-stationarity might be due to natural phenomena like the Doppler effect or artificial manipulations such as in orthogonal frequency division multiplexing (OFDM) \cite{madisetti_cyclostationary_2009, liu_channel_2014}.
In the audio processing domain, vowels are often modeled as an impulse train filtered by a time-varying linear filter. 
\Cref{fig:waveform} shows the waveform $x(t)$ of the /ä/ phoneme uttered by a male speaker, its power spectral density (PSD), and its bifrequency spectrum.
The bifrequency spectrum approximates $\E{X(f_1)X(f_2)^*}$ for all frequencies $f_1,f_2$, where $\E{\cdot}$ indicates the expected value and $X(f)$ is the Fourier transform of $x(t)$.
The vowel in \Cref{fig:waveform} has a non-diagonal bifrequency spectrum, implying that its frequency components are correlated.
First of all, this is not in line with the typical assumptions being made: estimation of parameters or processes from such an acoustic scene could be impaired.
Secondly, we can conclude that $x(t)$ cannot be modeled as a realization of a WSS process, and the ergodicity assumption does not hold \cite{wolfe_mechanics_2020, dong_characterizing_2017}.
Characterizing the spectral covariance of such a process requires a phase-adjusted estimator, whose details are discussed in this contribution as well.

Empirical studies on human auditory perception consistently highlight the practical importance of spectral correlations in spatial filtering.
These correlations are critical in tasks such as sound localization and speech intelligibility.
For instance, speech intelligibility in noise is influenced by the periodic structure of signals, with harmonically complex tones allowing for easier detection compared to inharmonic noise \cite{gockel_asymmetry_2002}.
Additionally, humans can localize speakers based on spatially aliased measurements, but only when spectrally complex sounds are present \cite{moore_introduction_2012, trahiotis_lateralization_1989}.
Dmochowsky \etal proved that spatial aliasing, a common issue in narrowband signals \cite{mccowan_2001_robust}, has reduced impact when the signals are wideband, regardless of the spatial sampling period \cite{dmochowski_spatial_2009}.
Despite the compelling evidence of the relevance of wideband patterns, traditional channel estimation algorithms have rarely considered them explicitly.

Therefore, this paper aims to investigate the impact of spectral correlations on the channel estimation task.
Our contributions are twofold:
Firstly, we propose an RTF estimation technique based on subspace analysis that exploits spectral and spatial correlations.
This technique consistently outperforms the narrowband maximum-likelihood estimator (MLE), known as covariance whitening (CW) \cite{markovich_multichannel_2009,markovich-golan_performance_2015,markovich-golan_performance_2018,varzandeh_iterative_2017}, when the target exhibits spectral correlations.
Secondly, we derive conditional and unconditional CRBs for the RTF estimation task.
To the best of our knowledge, bounds for the RTF estimation task have not been derived before, not even for the narrowband scenario.
The bounds show that channel estimation can be conducted more accurately if the target or the additive noise presents inter-frequency correlations.
Our findings align with experiments showing that both parametric methods and methods based on deep neural network (DNN) for speech enhancement, which jointly process spectral information, outperform their counterparts that process each frequency bin independently \cite{benesty_bifrequency_2012, huang_minimum_2014, tan_neural_2022, tesch_insights_2023}.
Although the accuracy of the proposed algorithm is generally close to the bound, there is some room for improvement, especially when noise signals with high spectral correlation are present.
\blue{An additional contribution is that, in the spirit of reproducible research, a Python implementation is freely available online\footnote{\url{https://github.com/Screeen/SVD-direct}}.}

The article details the signal model in \Cref{sec::sig_mod}. 
In \Cref{sec:estimation_cov_matrices}, we demonstrate how to recover the spectral-spatial covariance matrix of the source at the receivers, and introduce two related RTF estimation methods.
Based on these results, we propose a novel algorithm for RTF estimation in \Cref{sec::ch_est}. To better assess the algorithms' performance, we compare them to the lower bounds on the variance of RTF estimation, which are derived in \Cref{sec::crb}.
Numerical evidence of the superiority of the proposed algorithm, especially when the target presents spectral correlation, is provided in \Cref{sec:experiments}. 
In \Cref{sec:discussion}, we present additional discussion and insights on the experiments.
Finally, some conclusions are drawn in \Cref{sec:conclusions}, summarizing the essential findings and contributions of this paper.

\section{Signal model}\label{sec::sig_mod}\noindent
In a reverberant and noisy environment, we consider the case of a single point source impinging on an array of $M \geq 2$ sensors.
The signal received by the array is given in the STFT domain as:
\begin{align}\label{eq::sig_model}
\v{x}_k(l) &= \v{d}_k(l) + \v{v}_k(l) = s_k(l)\v{a}_k + \v{v}_k(l) \in \mathbb{C}^{M},
\end{align}
where $\v{d}_k(l) = {s}_k(l)\v{a}_k$ is the target signal at the receiver, $l=1,\dots,L$ is the time-frame index and the subscript $k=1,\dots,K$ denotes the frequency bin index.
The STFT coefficients of the target signal at the source are modeled by ${s}_k(l)$, which are realizations of complex random variables with zero mean.
The target coefficients are \emph{not} assumed to be mutually independent over frequency.
They can follow any probability distribution.
The transfer function $\v{a}_k \in \mathbb{C}^{M}$ models the wave propagation from the target point source to the $M$ sensors.
The transfer function is assumed to be an unknown deterministic quantity that typically needs to be estimated in beamforming applications.
The noise coefficients $\v{v}_k(l)$ are also modeled as complex random variables with zero mean and an arbitrary probability distribution.

Let us now consider the coefficients for all frequency components jointly.
Noisy coefficients corresponding to a single time frame $l$, for $M$ sensors, at $K$ frequencies, can be stacked in a column vector as in
$
\v{x} = \begin{bmatrix}
           \v{x}_1^T,
           \v{x}_2^T,
           \dots,
           \v{x}_K^T
         \end{bmatrix}^T \in \mathbb{C}^{KM}.
 $
The time-frame index $l$ is left out for notational convenience.
In a similar fashion, noise vectors $\v{v}_k$, transfer function vectors $\v{a}_k$ and desired signal $\v{d}_k$ can be stacked vertically to form $\v{v}$, $\v{a}$, and $\v{d}$, respectively, so that $\v{x} = \v{d} + \v{v}$.
In this case, it is helpful to collect the signal coefficients $s_k$ in a random vector
$
\v{\sbeforekron} = [s_1, s_2, \ldots, s_K]^T.
$
Let us also define
$
    \v{s} = \v{\sbeforekron} \kron \v{1}_M = [s_1 \v{1}_M^T, s_2 \v{1}_M^T, \ldots, s_K \v{1}_M^T]^T,
$
where $\kron$ is the Kronecker product and $\v{1}_M$ is the $M$-dimensional all-ones vector.
Next, let
\begin{equation}\label{eq:A_def}
\v{A}=\diagp{\v{a}}=\diag(a_{11}, \ldots, a_{1M}, a_{21}, \ldots a_{KM}),
\end{equation}
contain the transfer functions for all frequencies and sensors.
The vector of desired signals is then given by
\begin{align}\label{eq:desired_sig}
\v{d} = \v{A}\v{s} = \v{A}(\v{\sbeforekron} \kron \v{1}_M),
\end{align}
such that the noisy coefficients for the wideband model can be written as
\begin{equation}\label{eq::sig_model_wideband}
    \v{x} = \v{d} + \v{v} = \v{A}\v{s} + \v{v}. 
\end{equation}
\blue{
Notice that \Cref{eq::sig_model_wideband} generalizes the narrowband model with multiplicative transfer function (MTF) approximation (\Cref{eq::sig_model}) to a wideband scenario.
In the MTF approximation, the linear convolution in the time domain is represented as multiplication in the STFT domain \cite{gannot_consolidated_2017}.
This constrains the transfer functions $\v{a}_k$ to be at most $K$ samples long in the time domain, effectively capturing the early reflections only, and neglecting the late reverberation components.
}

Next, we model the spatial and spectral correlations between the signals.
Spatial correlation matrices are widely used in array processing to model relations between signals received at different sensors.
Here, we also consider \textit{spectral} correlations between different frequency components.
The spectral-spatial covariance matrix $\v{R}_x = \E{\v{x}\v{x}^H} \in \mathbb{C}^{KM \times KM}$, can be expressed as
\begin{align}
    &\v{R}_x =
    \begin{bmatrix}
    \v{r}_x(1, 1) & \v{r}_x(1,2) & \cdots & \v{r}_x(1,K)\\
    \v{r}_x(2, 1) & \v{r}_x(2, 2) & \cdots & \\
    \vdots & \vdots & \ddots & \vdots\\
    \v{r}_x(K, 1) & \v{r}_x(K,2) & \cdots & \v{r}_x(K,K)\\
  \end{bmatrix},
 \end{align}
where $(\cdot)^H$ indicates the conjugate transpose operation, and
$
    \v{r}_x(i, j) = \E{\v{x}_{i} \v{x}^H_{j}} \in \mathbb{C}^{M \times M}
$
is the spectral-spatial covariance matrix at two arbitrary frequencies $i,j$.
When noise and target signal are statistically uncorrelated, we have
$
    \v{R}_x = \v{R}_d + \v{R_v},
$
that is,
$
    \v{r}_x(i, j) =
    \E{s_i s_j^*}\v{a}_i\v{a}_j^H + \E{\v{v}_i\v{v}_j^H}.
$
Let us now introduce alternative formulations of the covariance matrices that will be useful for our analysis.
Using the definition in \Cref{eq:desired_sig}, the signal covariance matrix $\v{R}_d = \E{\v{d}\v{d}^H}$ 
can be expressed as
\begin{align}
\scalemath{1.}{
    \v{R}_d = \E{\v{A}\v{s}\v{s}^H\v{A}^H}
    = \v{A}\E{\v{s}\v{s}^H}\v{A}^H
    = \v{A}\v{R}_s\v{A}^H, \label{eq:desired_cov_receiver}
    }
\end{align}
where $\v{R}_s$ is defined as $\v{R}_s = \E{\v{s}\v{s}^H}$.
Using the properties of the Kronecker product, the covariance matrix $\v{R}_s$ can, in turn, be rewritten as
    \begin{align}
    \v{R}_s &= \E{(\v{\sbeforekron} \kron \v{1}_M)(\v{\sbeforekron} \kron \v{1}_M)^ H}
    = \E{\v{\sbeforekron}\v{\sbeforekron}^H \kron \v{1}_M\v{1}_M^H} \nonumber \\
    &= \E{\v{\sbeforekron}\v{\sbeforekron}^H} \kron \v{1}_{M\times M}
    = \v{R}_{\sbeforekron} \kron \v{1}_{M\times M}, \label{eq:rs_kron}
\end{align}
where $\v{1}_{M\times M}$ is the all-ones matrix of size ${M\times M}$ and
\begin{equation}\label{eq:rs_before_kron}
    \v{R}_{\sbeforekron} = \E{\v{\sbeforekron}\v{\sbeforekron}^H} \in \mathbb{C}^{K\times K}.
\end{equation}

\section{Background information}\label{sec:estimation_cov_matrices}\noindent
\def\figphasewidth{0.22\linewidth}
\begin{figure*}[t]
\vspace{-2mm}
    \centering
    \subfloat[]{%
    \includegraphics[width=\figphasewidth]{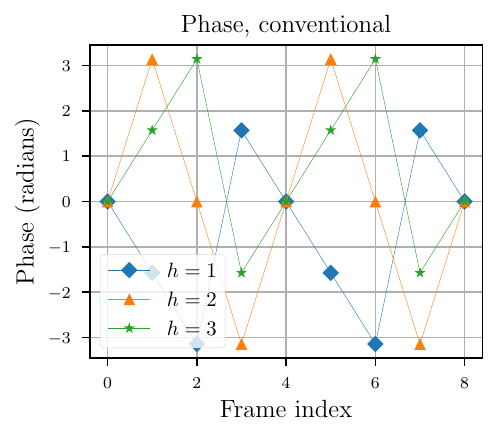}\label{fig:phase_corr_ph_orig}}
    \hfill%
    \subfloat[]{%
    \includegraphics[width=\figphasewidth]{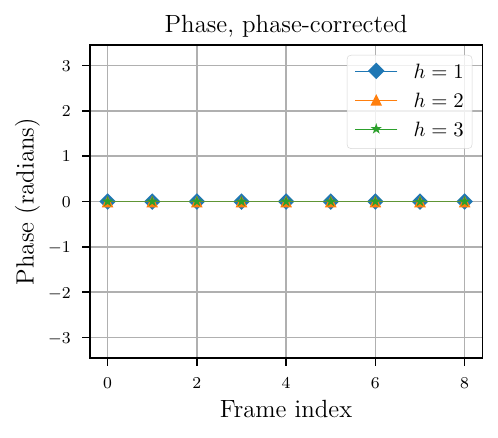}\label{fig:phase_corr_ph_corr}}
    \hfill%
    \subfloat[]{%
    \includegraphics[width=\figphasewidth]{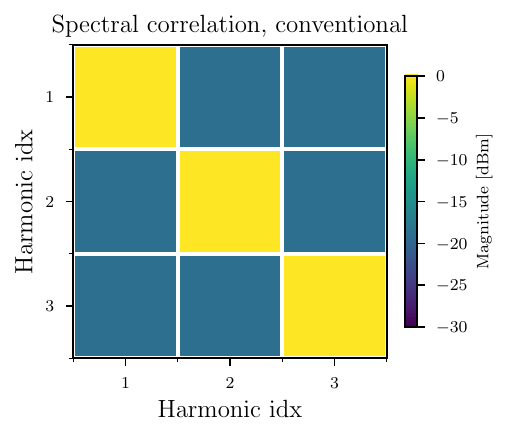}\label{fig:phase_corr_sc_orig}}
    \hfill%
    \subfloat[]{%
    \includegraphics[width=\figphasewidth]{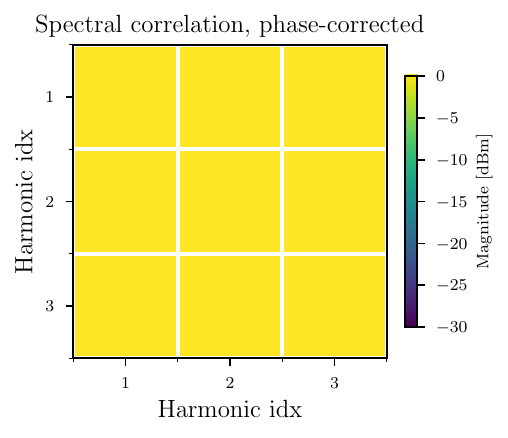}\label{fig:phase_corr_sc_corr}}
    \hfill%
  \caption{Phase correction improves estimation of the spectral covariance. (a) Phase components of the STFT of the original signal $y_k(l)$; (b) spectral covariance of $y_k(l)$; (c) phase components of the STFT of the phase-corrected signal $y_k^c(l)$; (d) spectral covariance of $y_k^c(l)$.}
  \label{fig:phase_corr}
  \vspace{-5mm}
\end{figure*}
This section begins by reporting a strategy to estimate sample spectral-spatial covariance matrices.
It then demonstrates that the desired signal covariance matrix $\v{R}_d$ is singular, with its rank being limited by the number of frequency components $K$.
\Cref{ssec:gevd} explores how the eigenvectors of $\v{R}_x$ and $\v{R}_d$ are affected by additive noise, and it reports a strategy for recovering $\v{R}_d$.
Finally, two well-known algorithms for RTF estimation are introduced.
\subsection{Phase-adjusted sample covariance matrix}\label{ssec:phase_adjust}\noindent
The commonly used sample covariance matrix estimate, serving as the MLE for jointly Gaussian WSS data, is expressed as 
\begin{equation}\label{eq:sample_cov_matrix}
    \tilde{\v{R}}_x = \frac{1}{L} \sum_{l=1}^{L}\v{x}(l) \v{x}(l)^H,
\end{equation} where $l$ is the realization index. 
Alternatively, $l$ can be treated as a time-frame index assuming second-order ergodicity.

However, when spectral correlations are present, the WSS assumption becomes inaccurate, requiring an alternative estimator for the spectral-spatial covariance matrices. 
In the estimation of \emph{spectral} correlations from STFT data, it is crucial to establish a connection among phase components across all frames and frequencies. 
In most implementations of the STFT, phase components are linked to the beginning of each frame. 
Therefore, there is a need to connect these phase components to a common reference point, such as the signal's onset, as mentioned by Antoni \cite{antoni_fast_2017}.
The phase-adjusted noisy STFT data at frequency $k$ is given by:
\begin{equation}\label{eq:phase_corr}
    \mathbf{x}_k^c(l) = \mathbf{x}_k(l) e^{-j 2 \pi l R k / K}, \quad l = 1, \ldots, L,
\end{equation}
where $R$ denotes the block shift between frames.

\blue{
Let us examine the impact of phase correction through an example. Consider a harmonic signal of the form $y(t) = \sum_{h=1}^3 \cos(2 \pi f_0 h t),~t \in \mathbb{N}$, where $f_0$ is the fundamental frequency in normalized units, and $h$ denotes the harmonic index.
The harmonic components at frequencies $f_0 h$ for $h=1,2,3$ are deterministic, thus perfectly correlated, meaning that knowing one component allows us to infer the value of another.
In the STFT domain, we denote the harmonic signal as $y_k(l),~l=1,\ldots,L$, and its phase-corrected version as $y_k^c(l)$.
The overlap of the STFT is set to $75\%$, corresponding to $R = K / 4$.
\Cref{fig:phase_corr_ph_orig} shows the phase components of the three non-zero frequency components of $y_k(l)$ across time frames.
Due to the misalignment between the block-shift $R$ and the periodicities of $y(t)$, the phases of the harmonics components appear to change randomly from frame to frame. 
However, after applying phase correction to get $y_k^c(l)$, we can accurately determine the phase of all components (\Cref{fig:phase_corr_ph_corr}) with respect to the time origin, $t=0$.
Let us also analyze the impact of phase correction on the estimation of the spectral correlations.
For the phase-corrected signal $y_k^c(l)$ (\Cref{fig:phase_corr_sc_corr}), the spectral correlation is maximal across all components, while the original $y_k(l)$ signal incorrectly appears to exhibit a lower spectral correlation due to spurious effects of phase cancellation (\Cref{fig:phase_corr_sc_orig}).
}

The phase correction becomes superfluous when dealing with products of components at the same frequency, as the conjugation leads to the cancellation of the phase term: $\mathbf{x}_k^c(l)\mathbf{x}_k^c(l)^H = \mathbf{x}_k(l)\mathbf{x}_k(l)^H$.
Similarly, the exponential term in \Cref{eq:phase_corr} is identical to one, thus ineffective, when $R=K$, \ie when adjacent frames do not overlap, or when independent realizations of the signals are used.
Therefore, the correction of \Cref{eq:phase_corr} is applied solely in \Cref{sec:real_speech_exp,sec:corr_analysis} to the overlapping STFT frames of real speech signals before covariance matrix estimation, so that, for $k_1, k_2 = 1, \ldots, K$,
\begin{align}
\label{eq:phase_corr_cov}
\begin{split}
 \hat{\v{r}}_x(k_1, k_2) &= \frac{1}{L}\sum_{l=1}^{L}\v{x}_{k_1}^c(l) \v{x}^c_{k_2}(l)^H \\
&= \frac{1}{L}\sum_{l=1}^{L}\v{x}_{k_1}(l) \v{x}_{k_2}(l)^H e^{-j 2 \pi l R (k_1 - k_2) / K}.
\end{split}
\end{align}

\subsection{Upper bound on the rank of target covariance matrix}
\begin{lemma}\label{lem::rank}
$\rank{\v{R}_d} \leq K$
\end{lemma}
\begin{proof}
To support this claim, we first state two well-known properties of the matrix rank.
Consider two matrices $\v{X} \in \mathbb{C}^{m\times n}$ and $\v{Y} \in \mathbb{C}^{n\times p}$.
According to \cite{golub_matrix_1983}, we have that:
\begin{gather}
    \rank{\v{X}\v{Y}} \leq \operatorname{min}({\rank{\v{X}}, \rank{\v{Y}}}), \label{eq:rank_product_prop} \\
    \rank{\v{X} \kron \v{Y}} = \rank{\v{X}} \rank{\v{Y}}. \label{eq:rank_kronecker}
\end{gather}
The covariance matrix $\v{R}_{\sbeforekron} = \E{\v{\sbeforekron}\v{\sbeforekron}^H}$ in \Cref{eq:rs_before_kron} obeys $\rank{\v{R}_{\sbeforekron}} \leq K$.
The rank of the all-one matrix, instead, is $\rank{\v{1}_{M\times M}} = 1$.
From \Cref{eq:rs_kron} and the rank property of Kronecker products in \Cref{eq:rank_kronecker} it follows that
\begin{align}\label{eq:rank_rs}
\scalemath{1.}{
    \rank{\v{R}_s} = \rank{\v{R}_{\sbeforekron} \kron \v{1}_{M\times M}} = \rank{\v{R}_{\sbeforekron}} \leq K.}
\end{align}
\blue{Moreover, let at least $K$ of the coefficients of the diagonal RTF matrix $\v{A}$ be non-zero by assumption, so that $\rank{\v{A}} \geq K$.}
It is now possible to analyze the matrix rank of $\v{R}_d$:
\begin{align}
    \rank{\v{R}_d} &=
    \rank{\v{A}\v{R}_s\v{A}^H} \\
    &\leq \operatorname{min}{\left(\rank{\v{A}}, \rank{\v{R}_s}\right)} 
    \leq K,
\end{align}
where the inequality follows from the rank matrix product property in \Cref{eq:rank_product_prop} and \Cref{eq:rank_rs}.
This completes the proof.
\end{proof}
\subsection{Estimation of the target covariance matrix}\label{ssec:gevd}\noindent
Suppose that $\v{R}_x$ is known, and 
\blue{let $\noisevar^2$ be the noise variance.
Estimated quantities are denoted as $\hat{(\cdot)}$. For example, the estimated noise variance is represented as $\hat{\noisevar}^2$.}
Assuming that the noise exhibits uniform power across both space and frequency, remaining uncorrelated in both domains, we have
$
\v{R}_v = \noisevar^2 \v{I}_{KM}.
$
As the identity matrix is diagonalizable by any unitary matrix,
\begin{equation*}
\scalemath{1.}{
    \v{R}_x = \v{R}_d + \noisevar^2 \v{I} = \v{V}\v{\Lambda}\v{V}^H + \noisevar^2 \v{I} = \v{V}(\v{\Lambda} + \noisevar^2 \v{I})\v{V}^H,
    }
\end{equation*}
where $\v{V}$ is the eigenvector matrix of $\v{R}_d$, and $\v{\Lambda}$ is the diagonal matrix containing the eigenvalues of $\v{R}_d$.
Therefore, if the estimated, phase-adjusted sample covariance matrix is decomposed as $\hat{\v{R}}_x = \hat{\v{V}}\hat{\v{\Lambda}}\hat{\v{V}}^H$, the covariance matrix of the target at the sensors can be approximated by $\hat{\v{R}}_d = \hat{\v{V}} \,\text{max}(\hat{\v{\Lambda}} - \hat{\noisevar}^2 \v{I}, 0) \hat{\v{V}}^H$,
where the $\text{max}(\cdot,\cdot)$ operator forces the eigenvalues of the Hermitian positive semidefinite (HPSD) matrix $\hat{\v{R}}_d$ to be non-negative.

If spatially or spectrally colored noise is present, the eigenvectors of $\v{R}_x$ and $\v{R}_d$ will differ.
However, estimating $\hat{\v{R}}_d$ and computing its eigenvalue decomposition is still possible if an estimate of the noise covariance matrix $\hat{\v{R}}_v$ is available and \blue{it is full-rank, hence invertible.
To ensure that this requirement is satisfied, we apply \emph{diagonal loading}, which consists of adding a scaled identity matrix to the estimated noise covariance matrix: $\hat{\v{R}}_v \leftarrow \hat{\v{R}}_v + \epsilon \v{I}$, where $\epsilon$ is a small positive value.
}
The clean covariance matrix $\hat{\v{R}}_d$ can be estimated from the generalized eigenvalue decomposition (GEVD) of $\hat{\v{R}}_x$ and $\hat{\v{R}}_v$ or from the eigenvalue decomposition of the prewhitened noisy covariance matrix $\hat{\v{R}}_v^{-1/2} \hat{\v{R}}_x \hat{\v{R}}_v^{-1/2}$.
The present examination will be limited to the GEVD because the two procedures are theoretically equivalent \cite{doclo_robust_2003, nadakuditi_fundamental_2010}\footnote{A standard routine for computing the GEVD of HPSD matrices is based on Cholesky decomposition \cite[Algorithm 8.7.1]{golub_matrix_1983}.
It is used in the popular LAPACK drivers \cite{lapack} that are the backbone of Matlab and Numpy/Scipy.}.
Given the estimates $\hat{\v{R}}_x$ and $\hat{\v{R}}_v$, an estimate of the desired covariance matrix $\hat{\v{R}}_d$ can be obtained as follows:
\begin{enumerate}
    \item Computation of $\hat{\v{R}}_x \v{U} = \hat{\v{R}}_v \v{U} \v{D}$ or, equivalently,
    $
        \v{Q}^H \hat{\v{R}}_x = \v{D} \v{Q}^H \hat{\v{R}}_v,
    $
    where $\v{D}$ are the generalized eigenvalues, $\v{U}$ are the right generalized eigenvectors, $\v{Q}$ are the left generalized eigenvectors, and $\v{U} = \v{Q}^{-H}$.
    \item Partitioning of the left eigenvectors $\v{Q} = [\v{Q}_x\ \v{Q}_v]$, where $\v{Q}_x$ comprises of the first $K_d$ columns of $\v{Q}$.
    \item Estimation of $\hat{\v{R}}_d$ as
    $\hat{\v{R}}_d = \v{Q}_x \,\text{max}(\v{D}_x - \v{I}, 0) \v{Q}_x^H$,
\end{enumerate}
\blue{
where $K_d$ is the estimated rank of $\v{R}_d$, and $\v{D}_x$ is a diagonal subblock formed by the first $K_d$ columns and rows of $\v{D}$.
By virtue of \Cref{lem::rank}, $K_d \leq K$.
The number of frames $L$ available for estimating the covariance matrices also constrains the maximum possible matrix rank, such that $K_d \leq L$.
As a consequence, in steps 2) and 3), $K_d = \min{(K, L)}$ eigenvalue-eigenvector pairs are retained.
Note that, due to the sparse spectral distribution of speech, the actual rank of $\v{R}_d$ might be lower than $K_d$. 
Specifically, since many frequency components of speech signals are zero, the corresponding rows and columns in $\v{R}_d$ will also be zero, reducing the rank of the matrix.
}

\blue{
\subsection{Covariance whitening and covariance subtraction}\label{sec::cw}\noindent
The GEVD routine detailed in \Cref{ssec:gevd} is also widely used in traditional, narrowband processing for estimating the target spatial covariance matrix. 
It is indeed at the core of the covariance whitening (CW) algorithm, one of the most effective techniques for RTF estimation \cite{markovich_multichannel_2009,markovich-golan_performance_2015,markovich-golan_performance_2018,varzandeh_iterative_2017}.
Let the (narrowband) noisy spatial covariance matrix be represented by $\v{R}_x(k) = \E{\v{x}_k(l) \v{x}_k(l)^H} \in \mathbb{C}^{M \times M}$ and the noise spatial covariance matrix by $\v{R}_v(k)$.
The CW technique consists of estimating the generalized left eigenvectors of $(\v{R}_x(k), \v{R}_v(k))$ for each discrete frequency $k$, and then retaining the eigenvector corresponding to the largest eigenvalue. 
Assuming that a single speaker is present, the rank of $\v{R}_d(k) =
\E{|s_k|^2}\v{a}_k\v{a}_k^H $ is 1. Therefore, the principal eigenvector equals the target RTF $\v{a}_k$ up to a multiplicative factor.
}

\blue{
Covariance subtraction (CS) is another popular technique for RTF estimation.
CS estimates the target spatial covariance matrix by subtracting the noise covariance matrix from the observed covariance matrix, \ie $\v{R}^{\text{CS}}_d(k) = \v{R}_x(k) - \v{R}_v(k)$.
The RTF is then estimated from the principal eigenvector of $\v{R}^{\text{CS}}_d(k)$.
This simpler technique is generally less accurate then CW \cite{markovich-golan_performance_2018}.
}
\section{Proposed RTF estimation algorithm: SVD-direct}\label{sec::ch_est}\noindent
In the preceding sections, the investigation focused on the spectral-spatial covariance matrix of a noisy signal received from multiple sensors.
The knowledge gained from this investigation can be applied to estimate the channel $\v{a}$, provided that estimates of the spectral-spatial covariance for both the noisy signal $\v{\hat{R}}_x$ and the noise-only signal $\v{\hat{R}}_v$ are available.
To this aim, we introduce a new method for RTF estimation.
The proposed algorithm is based on a row partitioning of the estimated spectral-spatial covariance $\v{\hat{R}}_d$, followed by an SVD on each frequency subblock.
The approach is named \textbf{SVD-direct} to emphasize the simplicity of its implementation and the central role played by the singular value decomposition.
\blue{
The proposed method extends the CW technique (\Cref{sec::cw}) to a wideband scenario, thus leveraging inter-frequency correlations for better estimation accuracy.
Unlike CW, multiple frequency components are processed simultaneously both in the prewhitening and in the ensuing decomposition step.
}

The basic idea of the proposed RTF estimation algorithm can be explained by an example. First, let us introduce a simplified case with $K=2$ frequency components, to gain some intuition on the structure of $\v{R}_d = \v{A}\v{R}_s\v{A}^H$.
We have that
\begin{align}
    \v{R}_d &=
    \begin{bmatrix}
    \E{|s_1|^2}\v{a}_1\v{a}_1^H & \E{s_1 s_2^*}\v{a}_1\v{a}_2^H \\
    \E{s_2 s_1^*}\v{a}_2\v{a}_1^H & \E{|s_2|^2}\v{a}_2\v{a}_2^H
    \end{bmatrix} =\\
  &=
  \begin{bmatrix}
  \sigma_{1}^2\v{a}_1\v{a}_1^H & \sigma_{12}    \v{a}_1\v{a}_2^H \\
    \sigma^*_{12}\v{a}_2\v{a}_1^H & \sigma_{2}^2\v{a}_2\v{a}_2^H
    \end{bmatrix}
    =
\begin{bmatrix}
    \v{R}_d^{(1)} \\
    \v{R}_d^{(2)}
\end{bmatrix}
    , \label{eq:rd1_rd2}
\end{align}
where we have introduced the auxiliary variables
$
\sigma_{1}^2 = \E{|s_1|^2},\  \sigma_{2}^2 = \E{|s_2|^2},\ \sigma_{12} = \E{s_1 s_2^*}
$
to simplify the notation.
The transfer function for the $i$th frequency is $\v{a}_i \in \mathbb{C}^{M}$.
We also defined the block-matrices
$
\v{R}_d^{(1)}, \v{R}_d^{(2)}  \in \mathbb{C}^{M\times2M}
$.
The absence of spectral correlations in the source signal $\v{s}$ would lead to $\E{s_1 s_2^*} = \E{s_2 s_1^*} = 0$.
Now, consider the block matrix
$
\v{R}_d^{(1)} =
\begin{bmatrix}
  \sigma_{1}^2\v{a}_1\v{a}_1^H & \sigma_{12}    \v{a}_1\v{a}_2^H
\end{bmatrix}
$
in \Cref{eq:rd1_rd2}.
Notice that $\v{R}_d^{(1)}$ is a rank-1 matrix, whose left principal singular vector is \blue{proportional to} $\v{a}_1$. %
The right principal singular vector of $\v{R}_d^{(1)}$ is proportional to $[\v{a}_1^T\ \v{a}_2^T]^T$.
To see this, consider the matrix product
\begin{align}
    \v{R}_d^{(1)}(\v{R}_d^{(1)})^H =
     (\sigma_{1}^2 \|\v{a}_1\|^2 + \sigma_{12}^2 \|\v{a}_2\|^2)\ \v{a}_1\v{a}_1^H 
\end{align}
from which it follows that $\v{R}_d^{(1)}(\v{R}_d^{(1)})^H$ is a rank-1 matrix with principal eigenvector $\v{a}_1$\footnote{Throughout the paper, $\|\cdot\|$ indicates the 2-norm.}.
It follows that by decomposing $\v{R}_d^{(1)}$ with an SVD and selecting the principal left singular component, $\v{a}_1$ can be recovered up to a scalar factor. 

The procedure above can be repeated for each subblock $\v{R}_d^{(k)} \in \mathbb{C}^{M\times KM},\ k=1,\dotsc, K$, leading to the proposed wideband channel estimation method, \textbf{SVD-direct} (\Cref{alg:svd-direct}).
\begin{algorithm}[t]
\caption{SVD-direct}\label{alg:svd-direct}
\begin{algorithmic}
\Require{$\hat{\v{R}}_x, \hat{\v{R}}_v, M, K$}
\Ensure{RTF estimates $\hat{\v{a}}$}
\texttt{\\}
\item \# Estimate $\hat{\v{R}}_d$ from the GEVD (\Cref{ssec:gevd}).
\State
$
     \hat{\v{R}}_d \gets
    \texttt{GEVD\_routine}
    (\hat{\v{R}}_x, \hat{\v{R}}_v)
 $
\texttt{\\}
\item{\# Partition in $K$ ``fat" $M\times KM$ blocks}
    \State
    $
     [(\hat{\v{R}}_{d}^{(1)})^T,
    (\hat{\v{R}}_{d}^{(2)})^T,
    \dots,
    (\hat{\v{R}}_{d}^{(K)})^T]^T \gets
    \hat{\v{R}}_{d}
    $
\texttt{\\}
\item \# {Per each frequency}
\For{$k = 1, \dotsc, K$}
    \State
    $
        \v{P}^{(k)}\v{D}^{(k)}\v{Q}^{(k)} \gets \texttt{SVD}( \hat{\v{R}}_{d}^{(k)})
    $
    \texttt{\\}
        \State \# {Rescale left principal singular vectors}
    \State $
    \v{\hat{a}}^{(k)} \gets \texttt{Normalize}({\v{p}_1^{(k)}}).
    $
\EndFor
\end{algorithmic}
\end{algorithm}
\noindent The function $\texttt{Normalize}$ is defined as $\texttt{Normalize}(\v{a}^{(k)}) = {\v{a}}^{(k)} / [{\v{a}}^{(k)}]_r$, and $[{\v{a}}^{(k)}]_r$ is the entry corresponding to the $r$-th (reference) sensor.

\section{Cramér--Rao lower bound}\label{sec::crb} \noindent
Based on the spectral-spatial covariance matrix of the signal received at the multiple sensors, we derived an algorithm for RTF estimation, taking correlation across frequency into account.
To determine how close this algorithm is to the optimal performance, we compare it to the CRB.

In the following, we first define the CRB and show how to derive it when estimating a deterministic function of an unknown parameter.
The CRB is then calculated for two scenarios: ($i$) a setting where the target signal $\v{s}(l)$ is deterministic and known (\Cref{sec:crb_cond}), and ($ii$) a scenario where the target signal has a known covariance matrix $\v{R}_s$, but the signal realizations are unknown (\Cref{sec:crb_uncond}).
Note that the former bound will lead to an unrealistic lower bound, as in the current scenario, $\v{s}(l)$ is never known.
The latter bound is realistic as it only assumes that the first- and second-order statistics are known.
The two settings are also known as the deterministic or conditional CRB, and stochastic or unconditional CRB, respectively \cite{stoica_performance_1990}.
Although the CRBs are derived for the wideband scenario, they encompass the bounds for narrowband RTF estimation as a specific case.

It is worth noting that the CRB for proper complex-valued multivariate Gaussian parameters has been previously explored. 
In \cite[Eq.\, 15.52]{kay_fundamentals_1993}, an approach that treats the real and imaginary components of the parameters independently was adopted. 
Conversely, in \cite[Eq.\, 6.55]{peter_j_schreier_statistical_2010}, the Wirtinger derivatives were employed.
However, neither of these references extends its analysis to incorporate further deterministic transformations.

\subsection{Problem formulation}\noindent
Let us consider the case where the parameters $\v{\theta}$ to be estimated are complex-valued, deterministic but unknown, and the observed data matrix is $\v{X} = [\v{x}(1) \, \ldots \, \v{x}(L)]$.
The distribution of the observed data is $p(\v{X}; \v{\theta})$.
The Fisher information matrix (FIM) is found as the negative expected Hessian of the log-likelihood function:
\begin{align}\label{eq::fim}
    \v{I}_{\v{\theta}} =
    - \E{\grad_{\v{\theta}}\grad^H_{\v{\theta}} \ln{p(\v{X}; \v{\theta})}}
    =
    - \E{\grad_{\v{\theta}}^2 \ln{p(\v{X}; \v{\theta})}},
\end{align}
where the expectation is taken with respect to $p(\v{X}; \v{\theta})$.
The gradient and the Hessian are defined as
\begin{equation*}
[\grad_{\v{\theta}}f]_i={\partial f}/{\partial \theta_i}, \qquad
[\grad_{\v{\theta}}^2 f]_{ij}={\partial^2 f}/{\partial \theta_i \partial \theta_j^*},
\end{equation*}
and the partial derivatives are Wirtinger derivatives \cite{brandwood_complex_1983}.
The covariance matrix $\v{R}_{\hat{\v{\theta}}}$ of any unbiased estimator $\hat{\v{\theta}}$ of ${\v{\theta}}$ satisfies $\v{R}_{\hat{\v{\theta}}} \succeq \v{I}_{\v{\theta}}^{-1}$\,\footnote{$\v{A} \succeq \v{B}$ means $\v{A} - \v{B}$ is positive semidefinite with $\v{A}$ and $\v{B}$ being Hermitian}.
When the quantity to estimate is given by a function $\v{\phi} = \v{g}(\v{\theta})$ of some underlying parameter, the bound follows as \cite{van_den_bos_cramer-rao_1994}
\begin{align}\label{eq:crb_det_function}
    \v{R}_{\hat{\v{\phi}}} \succeq (\grad_{\v{\theta}}\v{g})
    \v{I}^{-1}_{\v{\theta}}
    (\grad^H_{\v{\theta}}\v{g}),
\end{align}
where $\v{R}_{\hat{\v{\phi}}}$ is the covariance matrix of the estimator $\hat{\v{\phi}} = \v{g}(\hat{\v{\theta}})$.

In the present case, we define a function $\v{g} : \mathbb{C}^{2KM} \mapsto \mathbb{C}^{KM}$ that transforms a transfer function to a \textit{relative} transfer function.
It is given by
\begin{align}\label{eq:g_def}
\v{g}(\v{\theta}) &= \v{g}([\v{a}^T\ \v{a}^H]^T)
= \v{a} / \v{a}_{\text{ref}}, %
\end{align}
where the division is intended element-wise and
\begin{align*}
\v{a}_{\text{ref}} = [a_{1r} \v{1}_M^T, a_{2r} \v{1}_M^T, \ldots, a_{Kr} \v{1}_M^T]^T,
\end{align*}
is the vector with the responses of the $r$th (reference) sensor at all frequencies.
Notice that $\v{g}(\cdot)$ corresponds to the $\texttt{Normalize}(\cdot)$ function defined in \Cref{sec::ch_est}, with the only difference that $\v{g}(\cdot)$ acts on transfer functions for all frequencies and sensors simultaneously.
This function can be readily modified to accommodate various strategies for reference sensor selection \cite[Eq. 10]{gannot_consolidated_2017}.

\subsection{Conditional Cramér--Rao bound}\label{sec:crb_cond} \noindent
Consider the model from \Cref{eq::sig_model_wideband}:
\begin{align}\label{eq::sig_model_wideband_repeated}
\v{x}(l) = \v{A}\v{s}(l) + \v{v}(l),  %
\quad l = 1, \dots, L.
\end{align}
Firstly, we analyze the case where the signal $\v{s}(l)$ is known and the absolute transfer function $\v{A}$, defined in \Cref{eq:A_def},
is deterministic but unknown.
The noise $\v{v}(l)$ is a complex circular Gaussian random process with known spectral-spatial covariance $\v{R}_v$.
The vector of unknown parameters is $\v{\theta} = [\v{a}^T \v{a}^H]^T \in \mathbb{C}^{2KM}$.
The observed data $\v{X}$ follows a complex Gaussian distribution so that the log-likelihood is given by
\begin{align}\label{eq::likelihood_cond}
    \begin{split}
        &\ln{p(\v{X}; \v{\theta})} =
        -L\ln{|\pi\v{R}_v|} - {\sum_{l=1}^{L}\v{v}(l)^H \v{R}_v^{-1} \v{v}(l)}. \\
    \end{split}
\end{align}
We have the following result.
\begin{theorem}[Conditional CRB]\label{th:cond}
The variance of any conditional RTF estimator is lower bounded by:
\begin{align}\label{eq::crb_wideband_conditional}
    \text{CRB}[\v{g}(\hat{\v{\theta}})]_i =
    \left[
    (\grad_{\v{a}}\v{g})
    (\v{B}^*)^{-1}
    (\grad^H_{\v{a}}\v{g})
    \right]_{ii}
    ,
\end{align}
for $i=1,\dots,M$, where the matrix $\v{B}$ is defined as $\v{B} = \sum_{l=1}^{L} \v{S}(l)^H \v{R}_v^{-1} \v{S}(l)$ and $\v{S}(l) = \diagp{\v{s}(l)}$.
\end{theorem}
\begin{proof}
See \Cref{app:proof_crb_cond}.
\end{proof}
\blue{\subsubsection*{Interpretation}\noindent For ease of analysis, consider the case where the noise is spatially and spectrally uncorrelated, \ie $
\v{R}_v = \noisevar^2 \v{I}_{KM}
$. 
In this scenario, the $i$-th element on the diagonal of the Fisher information matrix is given by $[\v{I}_{\theta}]_{ii} = \noisevar^{-2} \sum_{l=1}^{L} |s_i(l)|^2$.
As the noise variance $\noisevar^2$ increases, the Fisher information $[I_{\theta}]_{ii}$ decreases. 
Conversely, increasing the number of frames $L$ available for estimation results in higher Fisher information, as the quantity $|s_i(l)|^2$ is always non-negative.
Thus, the achievable accuracy of the RTF estimation decreases with higher noise power and improves with more time frames.
}

\subsection{Unconditional Cramér--Rao bound} \label{sec:crb_uncond} \noindent
Consider again the model in \Cref{eq::sig_model_wideband_repeated}.
This time, we examine the more realistic scenario where the spectral-spatial covariance of the target signal $\v{R}_s$ is known but not the signal itself.
The transfer function $\v{A}$ is again deterministic but unknown.
This bound is then expected to be greater than the one derived in \Cref{th:cond} because the target signal is only known up to its second-order statistics.
In this case, the log-likelihood function is given by:
\begin{equation}\label{eq::likelihood_uncond}
\begin{split}
    \ln{p(\v{X}; \v{\theta})} &=
    -L\ln{|\pi\v{R}_x|} - L \trace{(\hat{\v{R}}_x\v{R}_x^{-1})},
\end{split}
\end{equation}
where $\v{R}_x = \v{A}\v{R}_s\v{A}^H + \v{R}_v$.
We have the following result.
\begin{theorem}[Unconditional CRB]\label{th:uncond}
In the unconditional settings, the variance of any unbiased RTF estimator is lower bounded by:
\begin{align}\label{eq::crb_wideband_unconditional}
    \text{CRB}[\v{g}(\hat{\v{\theta}})]_i =
    \left[
    (\grad_{\v{a}}\v{g})
    \v{C}
    (\grad^H_{\v{a}}\v{g})
    \right]_{ii},
\end{align}
for $i=1,\dots,M$, 
where $\v{C}$ is obtained by selecting the first $KM$ rows and columns from the inverse FIM $\v{I}_{\v{\theta}}^{-1}$.
\end{theorem}
\begin{proof}
See \Cref{app:proof_crb_uncond}.
\end{proof}
\blue{\subsubsection*{Interpretation}\noindent 
Consider the case where the noise is uncorrelated, \ie $
\v{R}_v = \noisevar^2 \v{I}_{KM}$. 
The $i$-th element on the diagonal of the Fisher information matrix is given by $[\v{I}_{\theta}]_{ii} = L \trace{\left(
            \v{R}_x^{-1} \v{F}_i \v{R}_x^{-1} \v{G}_i
            \right)}$.
As the number of frames $L$ available for estimation increases, the Fisher information increases linearly.
Thus, as we have seen for the conditional CRB in \Cref{sec:crb_cond}, the achievable accuracy of the RTF estimation improves with more time frames.
Also, as the noise variance $\noisevar^2$ increases, the Fisher information decreases, since $\v{R}_x^{-1} = (\v{R}_d + \v{R}_v)^{-1}$.
The numerical simulations in the following sections also reveal that the unconditional CRB is always equal to or higher than the conditional CRB.
Intuitively, when estimating the RTF, knowing the target signal itself would be more useful than knowing the signal statistics only.
For further analytical insights, the reader can refer to \cite{stoica_performance_1990}.
}

\section{Experiments}\label{sec:experiments}\noindent
In the preceding sections, we developed an RTF estimation algorithm that considers both spectral and spatial correlations. 
We computed conditional and unconditional CRBs to gauge achievable accuracy. 
Following this, we conduct simulations to compare the performance of our proposed wideband algorithm (SVD-direct) to the benchmark narrowband method (CW) and the established performance bounds.
We employ two error metrics, the root-mean-squared-error (RMSE) and the Hermitian angle \cite{varzandeh_iterative_2017}.
The RMSE is defined as:
\begin{align}\label{eq::rmse}
    \text{RMSE} =
    10 \log
    \sqrt{
    \frac{1}{KM}
        \|\hat{\v{a}} - \v{a}\|^2
        }
    \ (\si{\decibel})
    ,
\end{align}
while the Hermitian angle is calculated as:
\begin{align}\label{eq::hermitian_angle}
    \frac{1}{K}
        \sum_{k=1}^{K}
        \text{acos}{\left(
        \frac
        {|\,\hat{\v{a}}_k^H \v{a}_k|}
        {\|\hat{\v{a}}_k^H\| \|\v{a}_k\|}
        \right)}
        \ (\si{\radian})
        .
\end{align}
The RMSE accounts for discrepancies in the magnitude and phase, whereas the Hermitian angle depends exclusively on the angle between the RTFs.
The CRBs are only defined for error measures based on the MSE.
Therefore, these bounds are not shown in the plots that employ the Hermitian angle metric.
We also define the signal-to-noise ratio (SNR) in the frequency domain as:
\begin{align}\label{eq::snr}
    \text{SNR} =
    10 \log
    \frac
    {\sum_{i=1}^{KM}[\v{R}_d]_{ii}}
    {\sum_{i=1}^{KM}[\v{R}_v]_{ii}}
    \ (\si{\decibel}).
\end{align}
In all plots of \Cref{sec::equicorrelated,sec:varcorrelated,sec:real_speech_exp}, points connected by a continuous {\color{black}red line} show the error for the proposed algorithm (\Cref{alg:svd-direct});
points connected by a {\color{black}blue dotted line} show errors for the benchmark algorithm (CW); points connected by a {\color{black}green dash-dotted line} show the conditional CRB (\Cref{th:cond});
points connected by a {\color{black}purple dashed line} show the unconditional CRB (\Cref{th:uncond}).

We conduct five sets of experiments to explore increasingly realistic scenarios.
In the first two sets of experiments (\Cref{sec::equicorrelated,sec:varcorrelated}), we analyze scenarios where independent realizations of the signals are drawn from ideal multivariate Gaussian distributions.
In \Cref{sec::equicorrelated}, the target and noise powers at all frequencies are set to the same value and then rescaled to the desired SNR.
\Cref{sec:varcorrelated} describes a more realistic scenario where target and noise powers vary across frequencies.
Results are shown for a single random draw of the target TF $\v{a}$ and of the actual covariance matrices $\v{R}_s$ and $\v{R}_v$ because the CRB is defined for specific parameter values.
Nonetheless, similar outcomes are observed for other realizations.
To simulate the complex channel vector $\v{a}$, we generate two uniformly distributed random vectors with values from -1 to 1 and use them for the real and imaginary parts.
For the synthetic data of \Cref{sec::equicorrelated,sec:varcorrelated}, the lines in the figures are the mean results averaged across 5000 Montecarlo realizations. The faded area represents the \SI{95}{\percent} confidence interval \cite{altman_standard_2005}.
The bounds are evaluated at the actual values of the parameters.

The other three sets of experiments deal with real data.
The covariance matrices are thus estimated from overlapping STFT frames using the phase-corrected estimator introduced in \Cref{ssec:phase_adjust}.
\Cref{sec:corr_analysis} investigates the correlation coefficients of measured speech signals.
The experiments of \Cref{sec:real_speech_exp,sec:beamforming_bf} apply the proposed algorithms to recorded anechoic speech convolved with real room impulse responses (RIRs), and evaluate both the RTF estimation accuracy and the effect of employing the estimated RTFs for beamforming.
The ground truth TF $\v{a}$ is computed as the discrete Fourier transform of the first $K$ samples of the RIR.
We perform 50 Montecarlo repetitions of the real-data experiments.
Gaussian noise at \SI{40}{\decibel} SNR is added to $\v{v}$ in all experiments to account for sensor noise and simultaneously improve numerical conditioning of the inverse of the noise covariance matrix $\v{R}_v$.
\blue{
We also measure the computational complexity of the algorithm in \Cref{sec:complexity}.
As mentioned in \Cref{sec:intro}, all the simulations are implemented in Python, and the code to generate all figures in the paper is freely available online}.

\subsection{Equicorrelated, equal powers}
\label{sec::equicorrelated}
\noindent
The `equicorrelated' formulation, also considered in \cite{kasasbeh_noise_2017}, assumes that the noise signal exhibits identical variances at all sensors and frequency components.
The target signal has unit variance at all frequency components.
The cross-expectations over different frequency components are $\upsilon_f$ for the noise and $\rho_f$ for the target.
Because the frequency correlations are non-zero, the covariance matrices $\v{R}_x$ and $\v{R}_v$ describe non-WSS processes.
Taking again the case of $M=2$ sensors and $K=2$ frequency components to simplify the exposition, the noise covariance matrix $\v{R}_{v}$ is given by:
\begin{align}\label{eq:equicorr_noise_corr}
\v{R}_{v} = \noisevar^2
    \begin{bmatrix}
        1 & 0 & \upsilon_f & 0 \\
         0 & 1 & 0 & \upsilon_f \\
         \upsilon_f^* & 0 & 1 & 0 \\
         0 & \upsilon_f^* & 0 & 1
    \end{bmatrix},
\end{align}
where $\upsilon_f \in [0, 1]$ and $\noisevar^2$ is scaled according to \Cref{eq::snr} to yield the desired SNR.
Similarly, the desired covariance matrix at the source $\v{R}_{s} =
\v{R}_{\sbeforekron} \kron \v{1}_{M\times M}$ is given by:
\begin{align}\label{eq:equicorr_target_corr}
\scalemath{1.}{
\v{R}_{s} =
    \begin{bmatrix}
        1 & \rho_f \\
        \rho_f^* & 1
    \end{bmatrix} \kron \v{1}_{M\times M}
    =
    \begin{bmatrix}
        1 & 1 & \rho_f & \rho_f \\
         1 & 1 & \rho_f & \rho_f \\
         \rho_f^* & \rho_f^* & 1 & 1 \\
         \rho_f^* & \rho_f^* & 1 & 1
    \end{bmatrix}.
}
\end{align}
The desired covariance matrix at the receivers follows from \Cref{eq:desired_cov_receiver}.
The stimuli $\v{s}(l)$ and $\v{v}(l)$, where $l$ is the realization index, are generated through affine transformations applied to $L$ independent and identically distributed realizations $\v{n}(l)$ of a white complex multivariate Gaussian distribution $\v{n} \sim \mathcal{CN}(\v{0}, \v{I})$. For example, $\v{s}(l) = \v{R}_s^{1/2}\v{n}(l)$, and this implies $\v{s}(l) \sim \mathcal{CN}(\v{0}, \v{R}_s)$.
The estimates of the target and noise covariance matrices are derived through the sample covariance estimator of \Cref{eq:sample_cov_matrix}.
The phase-corrected estimator in \Cref{eq:phase_corr_cov} is indeed superfluous when independent signal realizations are available.
Unless specified differently, the SNR is set to $\SI{-5}{\decibel}$ in all experiments.
The signal correlation is set to $\rho_f=0.25$, the noise correlation to $\upsilon_f=0.25$, the number of frames to compute the sample covariance matrices to $L=1000$, the number of sensors to $M=2$ and the FFT length to $K=5$.
The true noise covariance matrix $\v{R}_v$ is used in all algorithms, aligning with the CRB assumptions.
Nonetheless, we noticed similar results when estimating $\v{R}_v$ from a separate realization of the noise-only signal.
The algorithms and the bounds are tested by varying four independent parameters: noise correlation $\upsilon_f$, target correlation $\rho_f$, number of time frames $L$, and SNR.

\subsubsection*{Varying noise correlation \texorpdfstring{$\upsilon_f$}{vf}}
\label{sec:equicorr_noisecorr}
In the first experiment, we analyze the performance of the algorithms as the noise frequency correlation $\upsilon_f$ varies between 0 and 1 (\Cref{fig:equicorr_noise_corr}).
We generally observe that the RMSE and the Hermitian angle metrics follow similar trends.
Let us first consider the scenario where the target has low correlation ($\rho_f=0.25$), corresponding to \Cref{fig:equicorr_noise_corr_lowrmse,fig:equicorr_noise_corr_lowha}.
The two algorithms perform equally well when the noise correlation $\upsilon_f$ is low, while the proposed method shows improved accuracy for high values of $\upsilon_f$.
In other words, the SVD-direct algorithm can partially take advantage of increased noise correlation, while the benchmark algorithm cannot.
Now, consider the case where the target shows high correlation ($\rho_f=0.75$), corresponding to \Cref{fig:equicorr_noise_corr_highrmse,fig:equicorr_noise_corr_highha}.
The proposed method outperforms the benchmark for all values of $\upsilon_f$, with improvements of approximately \SI{3}{\decibel} in RMSE and \SI{0.02}{\radian} in Hermitian angle.
Examining the conditional and unconditional CRBs, we note that substantial accuracy improvements are achievable when the noise exhibits a high correlation.
\begin{figure}[t]
\vspace{-5mm}
    \centering
    \subfloat[]{%
       \includegraphics[width=\figwidthnormal]{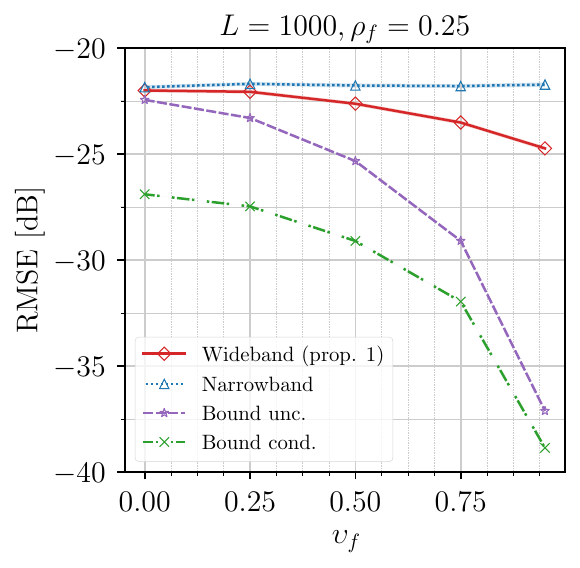}\label{fig:equicorr_noise_corr_lowrmse}}%
    \hfill%
    \subfloat[]{%
        \includegraphics[width=\figwidthnormal]{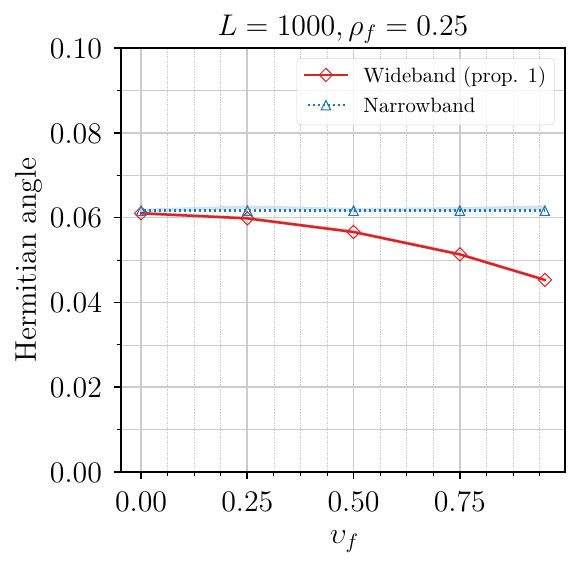}%
        \label{fig:equicorr_noise_corr_lowha}}%
    \hspace*{\fill}%
    \\
    \subfloat[]{%
        \includegraphics[width=\figwidthnormal]{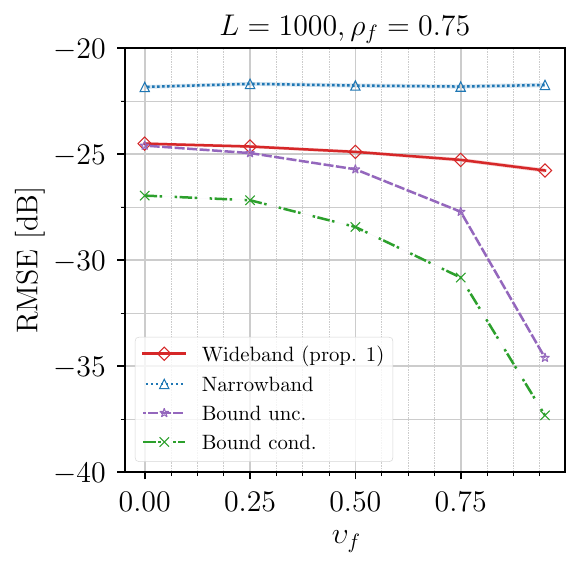}\label{fig:equicorr_noise_corr_highrmse}}
    \hfill
    \subfloat[]{%
       \includegraphics[width=\figwidthnormal]{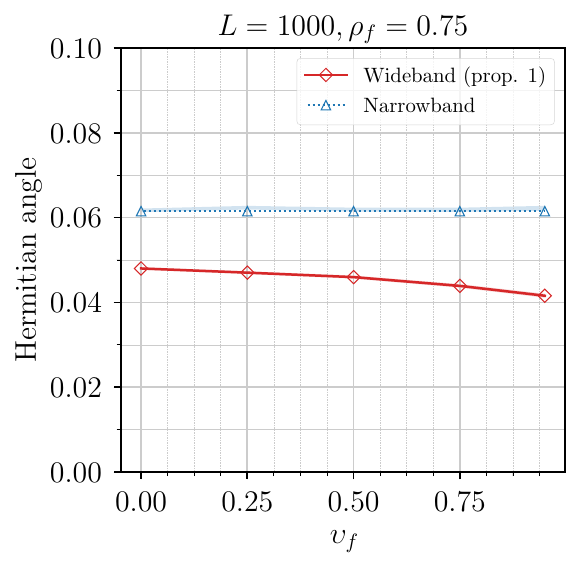}\label{fig:equicorr_noise_corr_highha}}
    \hfill
  \caption{Algorithm performance under varying noise frequency correlation $\upsilon_f$, with different levels of target correlation $\rho_f$. The top two plots (a) and (b) represent a less correlated target ($\rho_f=0.25$), while the bottom row (c) and (d) show a highly correlated target ($\rho_f=0.75$). Each column corresponds to different evaluation metrics: the left column displays the RMSE, and the right column shows the Hermitian angle.}
  \label{fig:equicorr_noise_corr}
\end{figure}
\subsubsection*{Varying target correlation \texorpdfstring{$\rho_f$}{pf}}\label{sec:equicorr_targetcorr}
In this section, we analyze the performance of the algorithms as the target frequency correlation $\rho_f$ varies between 0 and 1 (\Cref{fig:equicorr_target_corr}).
Because SVD-direct is explicitly designed to take advantage of spectral correlations in the target, we expect it to yield better performance for higher values of $\rho_f$.
If the noise has low correlation ($\upsilon_f=0.25$, corresponding to the top row in \Cref{fig:equicorr_target_corr}), the two algorithms perform equally well for low target correlation values $\rho_f$.
Additionally, we observe that the proposed method can fully exploit the target correlation and shows improvements in the accuracy of up to \SI{4}{\decibel} in RMSE and \SI{0.02}{\radian} in Hermitian angle for high values of $\rho_f$.
The benchmark algorithm is not affected by variations in the target spectral correlation.
Notice that the proposed algorithm achieves the CRB if a high target correlation is present, meaning that further improvements in accuracy in this scenario are not possible.
Interestingly, the unconditional performance bound exhibits different trends for low and high noise correlation. 
The unconditional bound decreases with higher target correlations when the noise correlation is low (\Cref{fig:equicorr_targetcorr_low_rmse}). 
Conversely, the maximum accuracy is lower as the target correlation increases for high noise correlation (\Cref{fig:equicorr_targetcorr_high_rmse}).
This aligns with findings from previous studies \cite{kasasbeh_noise_2017}.
This seeming discrepancy can be better understood through analogy: when two point sources are located close together in space, they show maximal \emph{spatial} correlation and exhibit similar correlation patterns, making it difficult to separate them. In our experiment, the noise and target sources have high \emph{spectral} correlation, and they share the same correlation pattern (\Cref{eq:equicorr_noise_corr,eq:equicorr_target_corr}).
We might say that they are ``spectrally superimposed" because their powers and correlation coefficients are the same, yielding very similar spectral covariance matrices.
As a result, they are harder to distinguish than two spectrally independent sources.
\begin{figure}[t]
\vspace{-5mm}
    \centering
    \subfloat[]{%
       \includegraphics[width=\figwidthnormal]{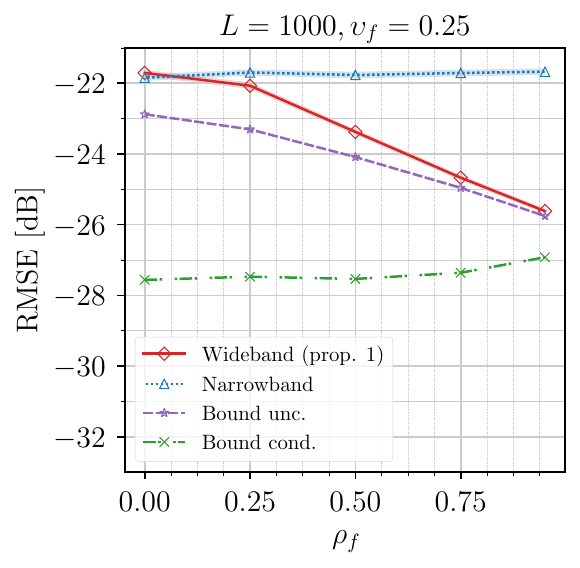}\label{fig:equicorr_targetcorr_low_rmse}
       }
    \hfill
    \subfloat[]{%
        \includegraphics[width=\figwidthnormal]{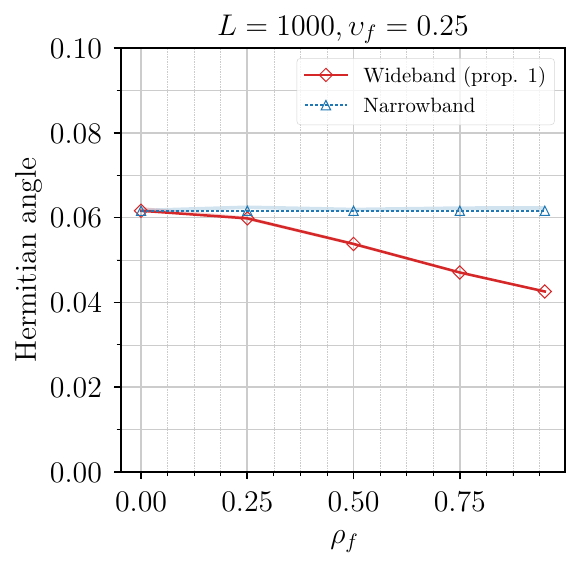}\label{fig:equicorr_targetcorr_low_ha}}
    \hfill
    \\
    \subfloat[]{%
        \includegraphics[width=\figwidthnormal]{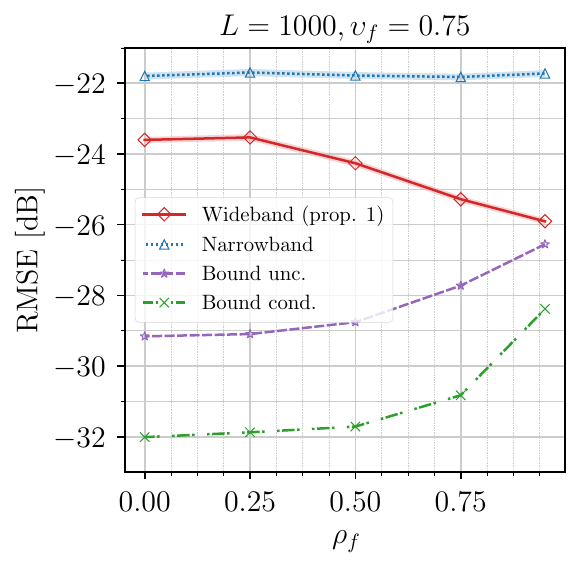}\label{fig:equicorr_targetcorr_high_rmse}}
    \hfill
    \subfloat[]{%
       \includegraphics[width=\figwidthnormal]{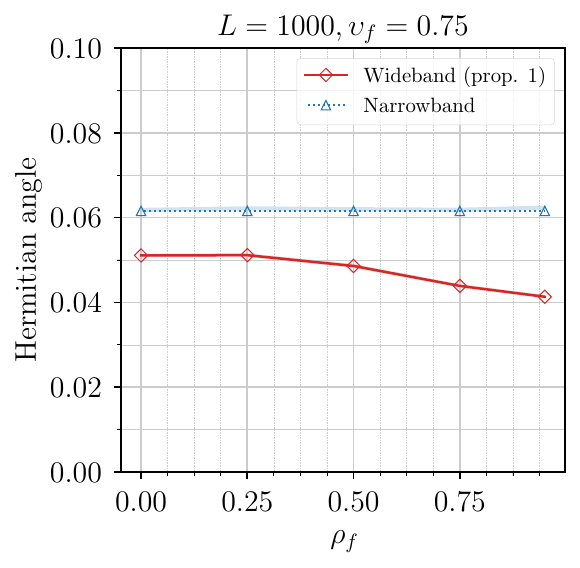}\label{fig:equicorr_targetcorr_high_ha}}
    \hfill
  \caption{Algorithm performance under varying target frequency correlation $\rho_f$, with different levels of noise correlation $\upsilon_f$. The top two plots (a) and (b) represent less correlated noise ($\upsilon_f=0.25$), while the bottom row (c) and (d) show highly correlated noise ($\upsilon_f=0.75$). The left column corresponds to RMSE, and the right column shows the Hermitian angle.}
  \label{fig:equicorr_target_corr}
\end{figure}

\subsubsection*{Varying number of frames \texorpdfstring{$L$}{L}}\label{sec:equicorr_frames}
We now analyze the performance of the algorithms when the number of frames $L$ to estimate the target covariance matrix ${\v{R}}_d$ is varied between $L=10$ and $L=5000$ (\Cref{fig:equicorr_frames}).
As expected, both algorithms perform better when more frames are available.
For low values of target and noise correlation (\Cref{fig:equicorr_frames_lowrmse,fig:equicorr_frames_lowherm}), the two algorithms perform similarly when the number of available time frames is large, whereas the proposed algorithm is slightly less accurate when $L$ is small.
When the target correlation is high, the wideband method shows smaller errors for any number of frames $L>10$ and converges to the unconditional CRB for a high number of frames.
\begin{figure}[t]
\vspace{-5mm}
    \centering
    \subfloat[]{%
       \includegraphics[width=\figwidthnormal]{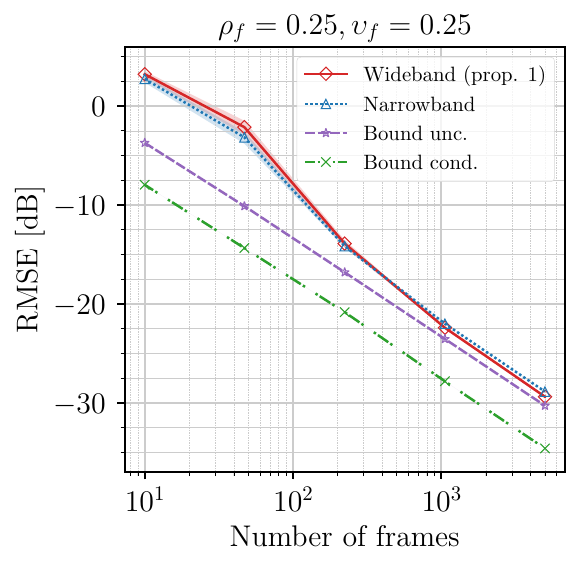}\label{fig:equicorr_frames_lowrmse}
       }
    \hfill
    \subfloat[]{%
        \includegraphics[width=\figwidthnormal]{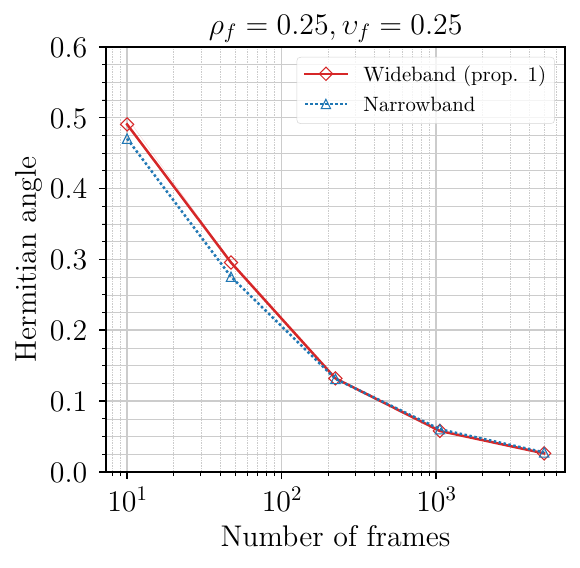}\label{fig:equicorr_frames_lowherm}
        }
    \hfill
    \\
    \subfloat[]{%
        \includegraphics[width=\figwidthnormal]{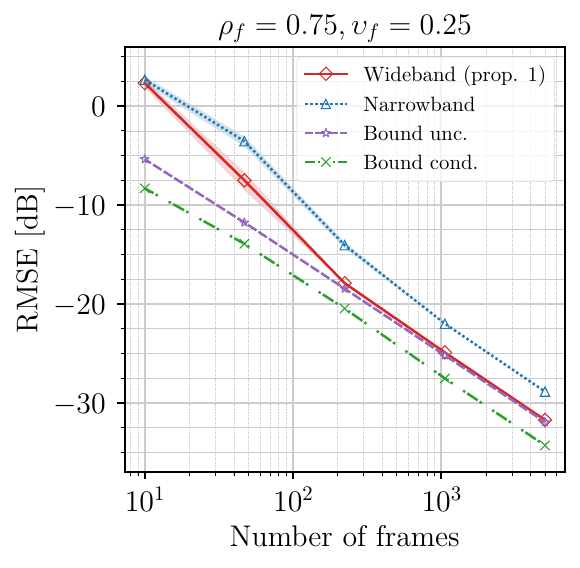}}
    \hfill
    \subfloat[]{%
       \includegraphics[width=\figwidthnormal]{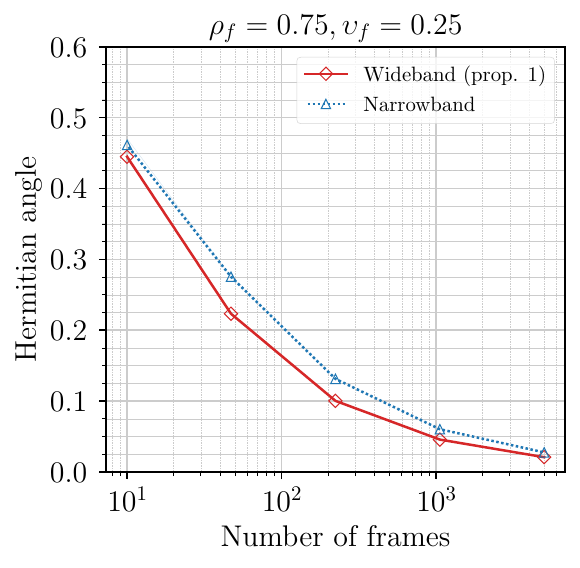}}
    \hfill
  \caption{Algorithm performance under varying number of time frames $L$, with different levels of target correlation $\rho_f$. The top two plots (a) and (b) represent a less correlated target ($\rho_f=0.25$), while the bottom row (c) and (d) show a highly correlated target ($\rho_f=0.75$). The left column corresponds to RMSE, and the right column shows the Hermitian angle.}
  \label{fig:equicorr_frames}
\end{figure}
\subsubsection*{Varying SNR}\label{sec:equicorr_snr}
\begin{figure}[t]
\vspace{-3mm}
    \centering
    \subfloat[]{%
       \includegraphics[width=\figwidthnormal]{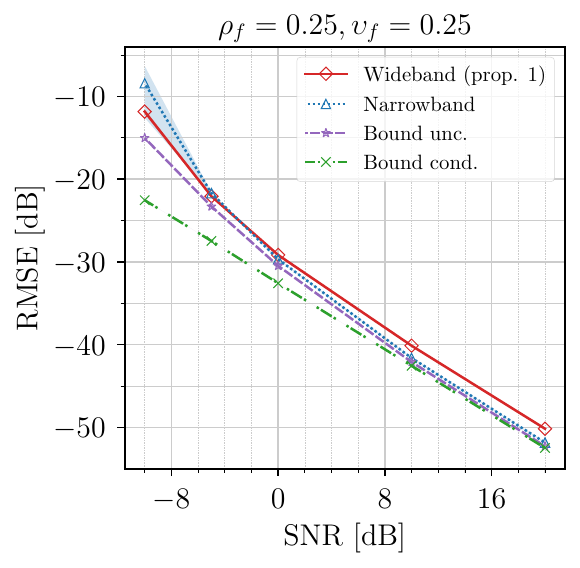}\label{fig:equicorr_snr_lowlow_rmse}
       }
    \hfill
    \subfloat[]{%
        \includegraphics[width=\figwidthnormal]{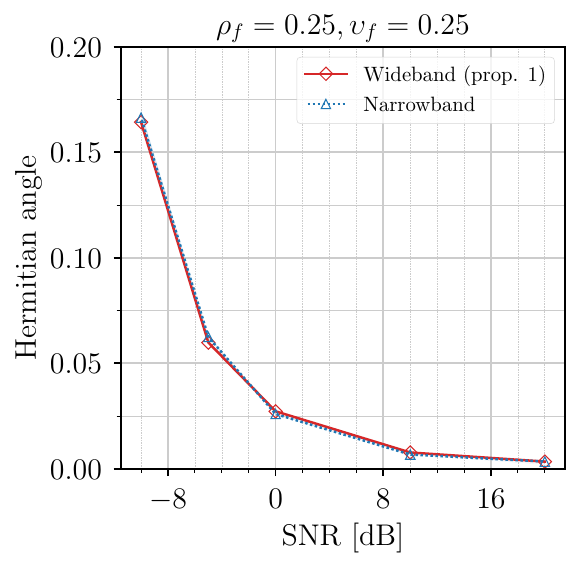}\label{fig:equicorr_snr_lowlow_ha}}
    \hfill
    \\
    \subfloat[]{%
        \includegraphics[width=\figwidthnormal]{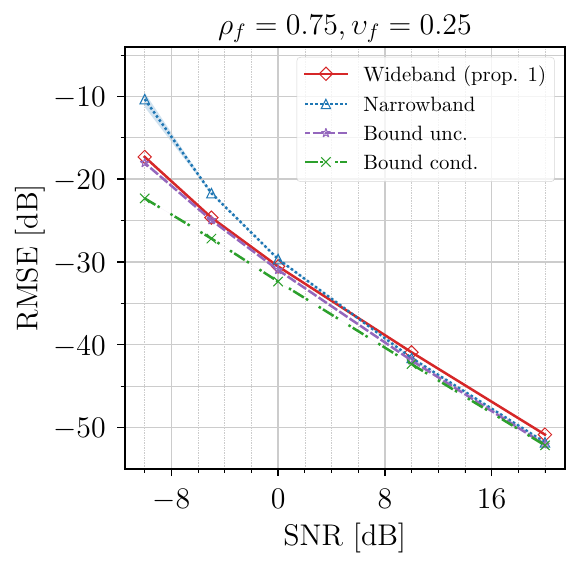}\label{fig:equicorr_snr_high_rmse}}
    \hfill
    \subfloat[]{%
       \includegraphics[width=\figwidthnormal]{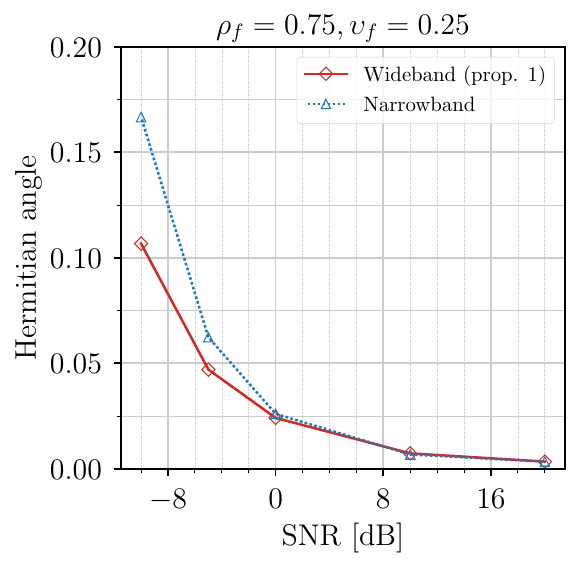}\label{fig:equicorr_snr_high_ha}}
    \hfill
  \caption{Algorithm performance under varying SNR, with different levels of target correlation $\rho_f$. The top two plots (a) and (b) represent a less correlated target ($\rho_f=0.25$), while the bottom row (c) and (d) show a highly correlated target ($\rho_f=0.75$). The left column corresponds to RMSE, and the right column shows the Hermitian angle.}
  \label{fig:equicorr_snr}
\end{figure}
This experiment analyzes the performance of the algorithms when the SNR varies between \SI{-10} and \SI{20}{\decibel} (\Cref{fig:equicorr_snr}).
Unsurprisingly, both algorithms perform better when the noise is less prominent.
The methods perform similarly for low target correlation values (\Cref{fig:equicorr_snr_lowlow_rmse,fig:equicorr_snr_lowlow_ha}).
In contrast, for high target correlation (\Cref{fig:equicorr_snr_high_rmse,fig:equicorr_snr_high_ha}), the proposed method shows significant performance gains of up \SI{8}{\decibel} in RMSE and \SI{0.05}{\radian} in Hermitian angle in noisy scenarios.
Both algorithms are close to the unconditional CRB for high SNR values.
\subsection{Equicorrelated, different powers}\label{sec:varcorrelated}\noindent
In the second set of experiments, we extend the `equicorrelated' scenario described in \Cref{sec::equicorrelated}, by incorporating varying signal powers across different frequency components and sensors.
This scenario is not only more realistic than the previous one, but it also leads to more diverse spectral correlation patterns --- that is, covariance matrices --- for the target and the noise signal, limiting the ``spectral superposition" phenomenon observed in \Cref{sec:equicorr_targetcorr}.
In this simulation model, special care must be taken to ensure the validity of the simulated covariance matrices $\v{R}_v$ and $\v{R}_s$.

Let $[\v{v}]_{km} = v_{km}$ be the noise signal at frequency $k$ and sensor $m$,
with variance $\E{ |v_{k m}|^2} = \noisevar{_{k m}^2}$.
By the Cauchy–Schwarz inequality, it is known that the covariance between the two discrete complex random variables $v_{k_1 m_1},v_{k_2 m_2}$, corresponding to the $(k_1 m_1, k_2 m_2)$ element of $\v{R}_v$, is upper-bounded by:
\begin{align}\label{eq:cauchy-rv-noise}
     |\E{ v_{k_1 m_1} v_{k_2 m_2}^*}|^2
     \leq
     \E{ |v_{k_1 m_1}|^2} \E{ |v_{k_2 m_2}|^2}.
\end{align}
Therefore, we can simulate $\v{R}_v$ with a two-step procedure.
First, the variances $\noisevar^2_{km}$ on the diagonal of $\v{R}_v$ are drawn from a uniform distribution $\mathcal{U}(\epsilon, 0.5)$, where $\epsilon > 0$ is a small positive number.
Next, the covariance values are calculated as
\begin{align}\label{eq:varcorrelated_crosselement}
     \E{ v_{k_1 m_1} v_{k_2 m_2}^*} =
     \begin{cases}
     0, & \text{if}\ m_1 \neq m_2, \\
     \upsilon_f \sqrt{\noisevar{_{k_1 m_1}^2} \noisevar{_{k_2 m_2}^2}} & \text{if}\ m_1 = m_2, \\
     \end{cases}
\end{align}
where the factor $\upsilon_f \in [0, 1]$ models the noise inter-frequency correlation.
Because $\upsilon_f \leq 1$, \Cref{eq:varcorrelated_crosselement} leads to covariance values that are always smaller than their theoretical maxima.
The correlations across different sensors are set to 0 since we model spatially uncorrelated noise.
Finally, $\v{R}_v$ is rescaled by a global noise variance $\noisevar^2$ to yield the desired SNR according to \Cref{eq::snr}.
Analogously, the desired covariance matrix at the source is given by:
\begin{align}
[\v{R}_{\sbeforekron}]_{k_1 k_2} =
    \begin{dcases}
        \sigma^2_{k_1 k_2} \sim \mathcal{U}(\epsilon, 0.5) &\text{if } k_1 = k_2, \\
        \rho_f \sqrt{\sigma_{k_1 k_1}^2 \sigma_{k_2 k_2}^2} &\text{if } k_1 \neq k_2.
    \end{dcases}
\end{align}
The desired covariance matrix at the receivers follows again from \Cref{eq:desired_cov_receiver,eq:rs_kron}.
The sampling procedure and the other simulation parameters follow from \Cref{sec::equicorrelated}.
\subsubsection*{Varying noise correlation \texorpdfstring{$\upsilon_f$}{vf}}
The performance of the algorithms is examined as the noise frequency correlation  $\upsilon_f$ varies from 0 to 1 (\Cref{fig:varcor_noise_corr}).
The wideband algorithm outperforms the narrowband one in all cases.
The difference in accuracy is larger when the target correlation $\rho_f$ is higher (\Cref{fig:varcor_noise_corr-highrmse,fig:varcor_noise_corr-highha}), reaching an improvement of up to \SI{8}{\decibel} RMSE and \SI{0.05}{\radian}.
Notice that the performance gains are more significant than in the `equal powers' scenario of \Cref{sec:equicorr_noisecorr}.
We also observe that the error of the SVD-direct algorithm slightly decreases for very high noise correlation $\upsilon_f \geq 0.75$. Still, the gap between the unconditional CRB and the algorithms indicates that further improvements are possible.
\begin{figure}[tbp]
\vspace{-5mm}
    \centering
    \subfloat[]{%
       \includegraphics[width=\figwidthnormal]{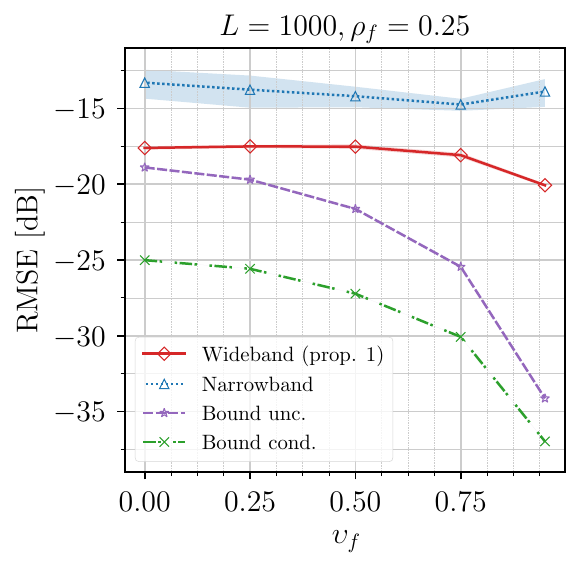}\label{fig:varcor_noise_corr-lowrmse}
       }
    \hfill
    \subfloat[]{%
        \includegraphics[width=\figwidthnormal]{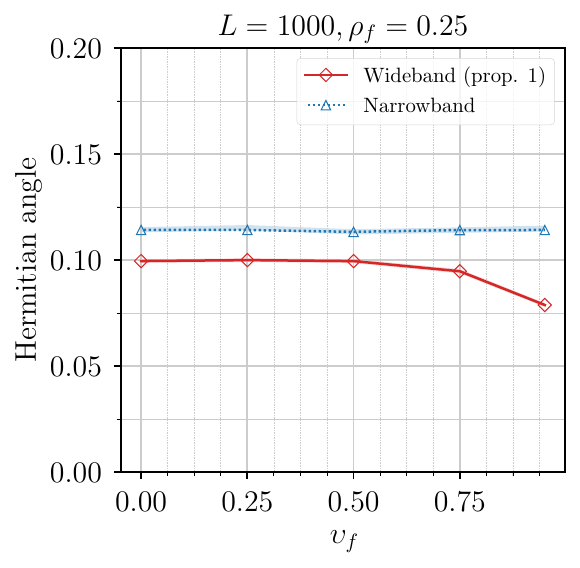}\label{fig:varcor_noise_corr-lowha}}
    \hfill
    \\
    \subfloat[]{%
        \includegraphics[width=\figwidthnormal]{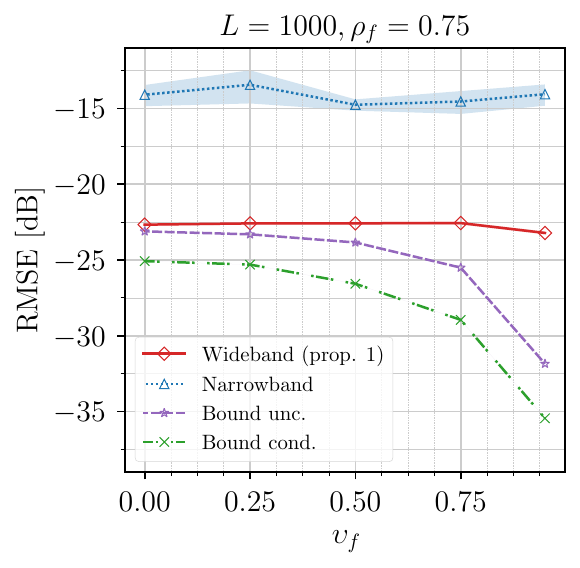}\label{fig:varcor_noise_corr-highrmse}}
    \hfill
    \subfloat[]{%
       \includegraphics[width=\figwidthnormal]{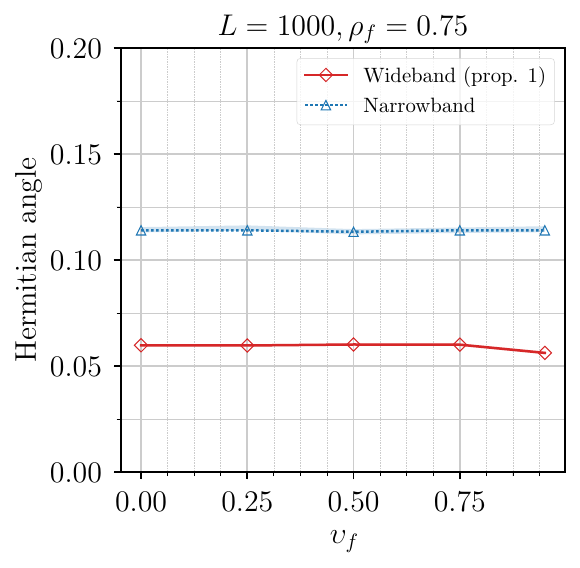}\label{fig:varcor_noise_corr-highha}}
    \hfill
  \caption{Algorithm performance for non-uniform target and noise powers under varying noise frequency correlation $\upsilon_f$, with different levels of target correlation $\rho_f$.
  The top two plots (a) and (b) represent a less correlated target ($\rho_f=0.25$), while the bottom row (c) and (d) show a highly correlated target ($\rho_f=0.75$). The left column corresponds to RMSE, and the right column shows the Hermitian angle.}
  \label{fig:varcor_noise_corr}
\end{figure}
\subsubsection*{Varying target correlation \texorpdfstring{$\rho_f$}{pf}}
Next, we turn to one of the key experiments of the present study, where we analyze the performance of the algorithms as the target frequency correlation $\rho_f$ varies between 0 and 1 for arbitrary noise and signal powers (\Cref{fig:varcor_target_corr}).
The wideband algorithm takes advantage of higher target spectral correlations $\rho_f$, as already observed in \Cref{sec:equicorr_targetcorr}: both for low and high noise correlation, SVD-direct has significantly better performance for higher values of $\rho_f$, reaching improvements of \SI{10}{\decibel} RMSE and \SI{0.05}{\radian} Hermitian angle.
On the other hand, CW performs slightly better when the target correlation is completely absent ($\rho_f = 0$).
A likely explanation is that the narrowband approach exploits the a priori knowledge that the target signal is uncorrelated across frequency.
However, we argue that the scenario where $\rho_f=0$ is unlikely to occur in practice for the reasons highlighted in \Cref{sec:intro}.
\blue{This intuition is also confirmed in the correlation analysis of real data in \Cref{sec:corr_analysis}.}
Turning our focus to the performance bounds, we notice that the SVD-direct method achieves the CRB if a high target correlation is present.
\hide{Similar to \Cref{sec:equicorr_targetcorr}, when the noise correlation is lower (\Cref{fig:varcor_target_corr_lowrmse}), the unconditional CRB decreases for increasing target correlation.
In contrast, for higher noise correlation (\Cref{fig:varcor_target_corr_highrmse}), the interaction between noise and target correlation leads to an unconditional CRB that is almost constant for varying $\rho_f$. Notice that such interactions were also observed in \cite{kasasbeh_noise_2017}.}
\begin{figure}[tbp]
\vspace{-5mm}
    \centering
    \subfloat[]{%
       \includegraphics[width=\figwidthnormal]{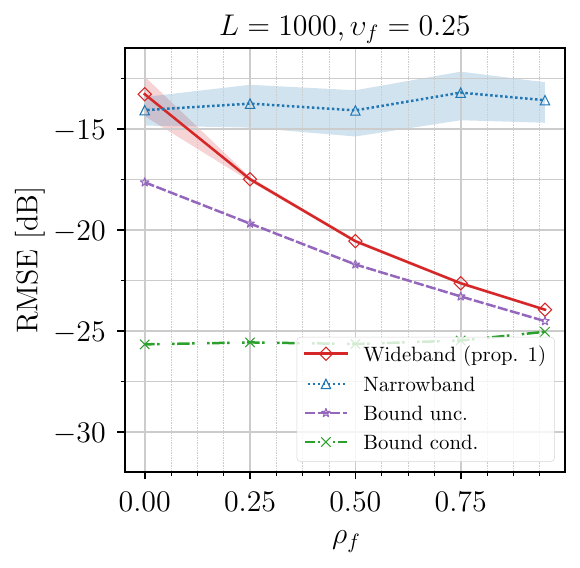}\label{fig:varcor_target_corr_lowrmse}
       }
    \hfill
    \subfloat[]{%
        \includegraphics[width=\figwidthnormal]{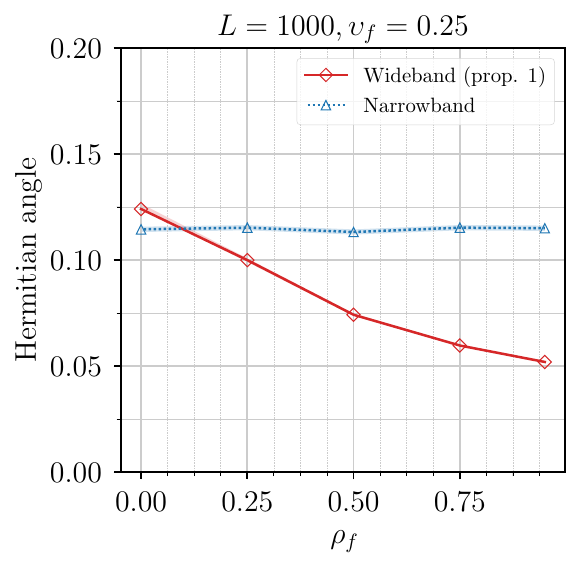}}
    \hfill
    \\
    \subfloat[]{%
        \includegraphics[width=\figwidthnormal]{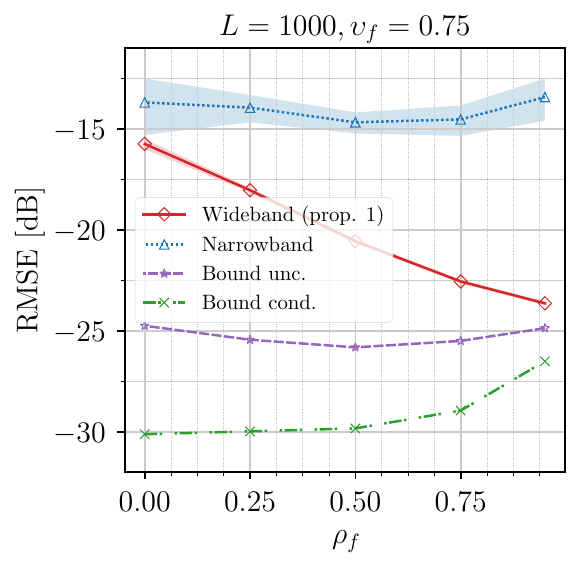}\label{fig:varcor_target_corr_highrmse}}
    \hfill
    \subfloat[]{%
       \includegraphics[width=\figwidthnormal]{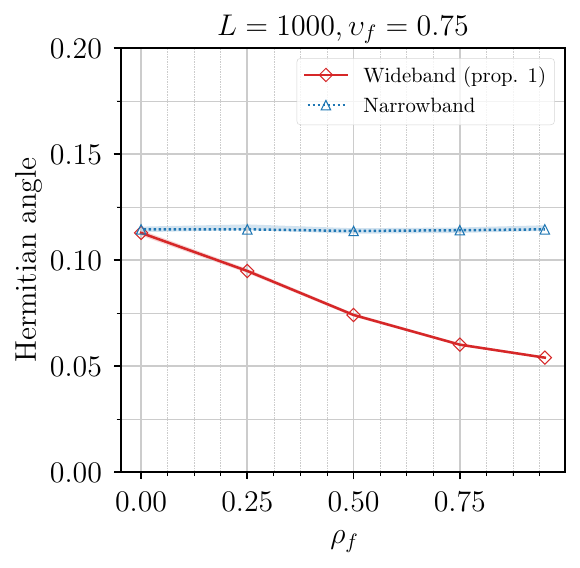}}
    \hfill
  \caption{Algorithm performance for non-uniform target and noise powers, under varying target frequency correlation $\rho_f$, with different levels of noise correlation $\upsilon_f$. The top row ((a) and (b)) represents less correlated noise ($\upsilon_f=0.25$), while the bottom row ((c) and (d)) shows highly correlated noise ($\upsilon_f=0.75$). The left column corresponds to RMSE, and the right column shows the Hermitian angle.}
  \label{fig:varcor_target_corr}
\end{figure}
\subsubsection*{Varying number of frames \texorpdfstring{$L$}{L}}
We now evaluate the different approaches when estimating the covariance matrices with varying numbers of time frames $L$ (\Cref{fig:varcor_frames}).
The narrowband and wideband approaches exhibit similar performance for scenarios with low target and noise correlation (\Cref{fig:varcor_frames_lowrmse,fig:varcor_frames_lowha}).
The benchmark method is slightly more accurate when only a few frames are available ($L<50$), whereas the proposed algorithm outperforms CW for a higher number of frames.
Because wideband spectral-spatial covariance matrices are considerably larger than narrowband spatial covariance matrices, accurate estimation of the former requires more realizations.
When the target is highly correlated (\Cref{fig:varcor_frames_highrmse,fig:varcor_frames_highha}), the wideband method consistently matches or outperforms the narrowband method.
\begin{figure}[t]
\vspace{-5mm}
    \centering
    \subfloat[]{%
       \includegraphics[width=\figwidthnormal]{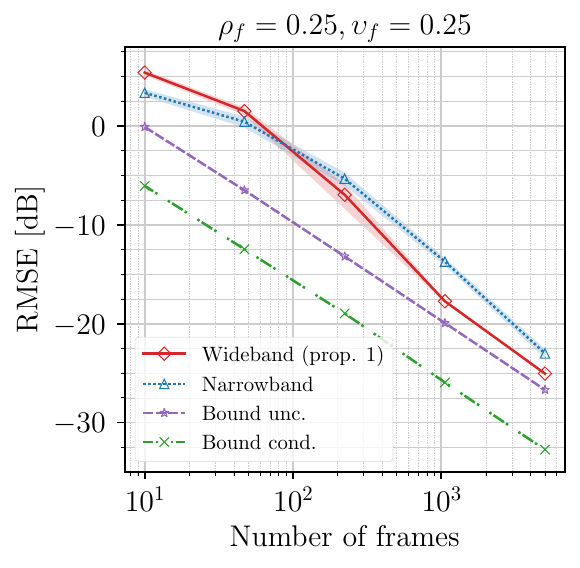}\label{fig:varcor_frames_lowrmse}
       }
    \hfill
    \subfloat[]{%
        \includegraphics[width=\figwidthnormal]{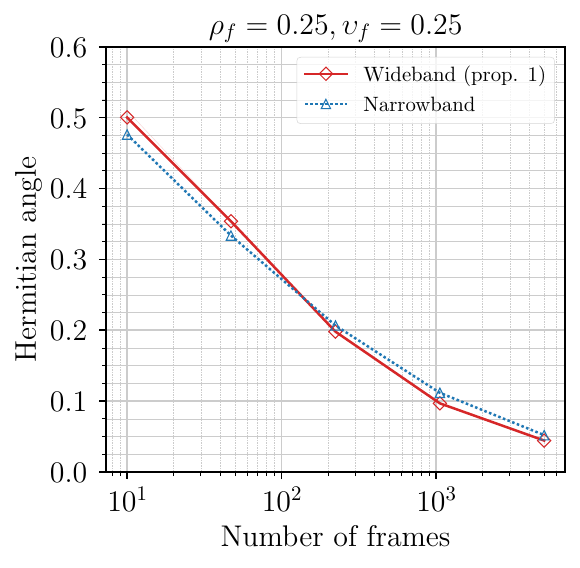}\label{fig:varcor_frames_lowha}}
    \hfill
    \\
    \subfloat[]{%
        \includegraphics[width=\figwidthnormal]{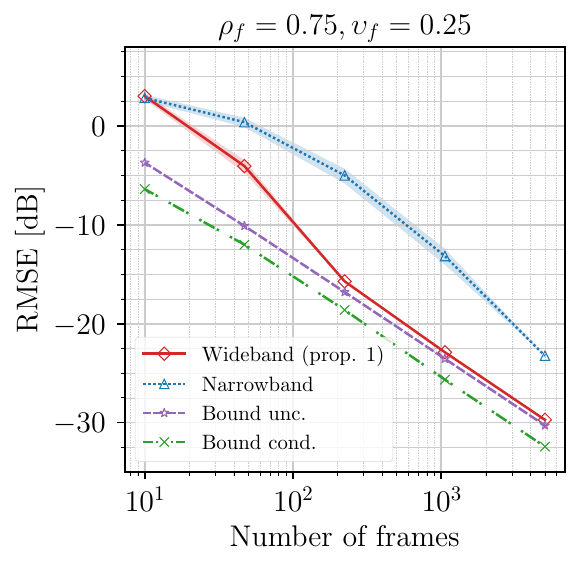}\label{fig:varcor_frames_highrmse}}
    \hfill
    \subfloat[]{%
       \includegraphics[width=\figwidthnormal]{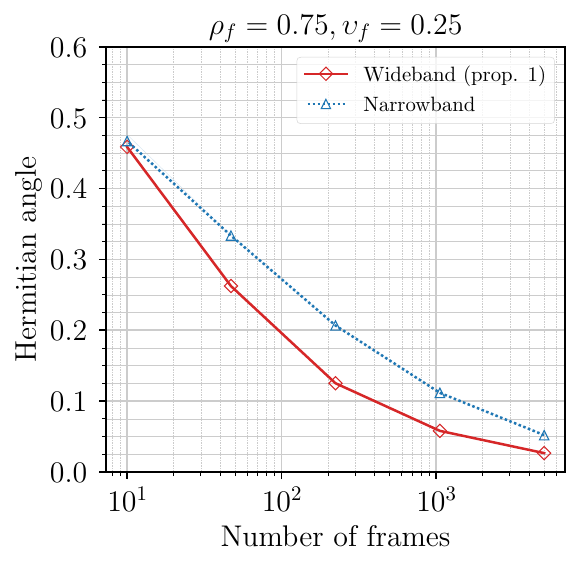}\label{fig:varcor_frames_highha}}
    \hfill
  \caption{Algorithm performance for non-uniform target and noise powers, under varying number of time frames $L$, with different levels of target correlation $\rho_f$. The top row ((a) and (b)) represents less correlated target ($\rho_f=0.25$), while the bottom row ((c) and (d)) shows highly correlated target ($\rho_f=0.75$). The left column corresponds to RMSE, and the right column shows the Hermitian angle.}
  \label{fig:varcor_frames}
\end{figure}
\subsubsection*{Varying SNR}
Lastly, in \Cref{fig:varcor_snr}, we examine the performance for various SNR levels.
For low target and noise spectral correlation (\Cref{fig:varcor_snr_lowrmse,fig:varcor_snr_lowha}), the two algorithms perform comparably, with the wideband method being marginally less accurate for higher SNRs.
By contrast, when the target correlation is high, as in \Cref{fig:varcor_snr_highrmse,fig:varcor_snr_highha}, the two approaches perform similarly for less noisy scenarios, but the wideband method has a considerably lower error for lower SNRs, with a reduction of up to \SI{11}{\decibel} RMSE and \SI{0.12}{\radian} Hermitian angle at \SI{-10}{\decibel} SNR.
This experiment concludes the evaluations on synthetic data.
\begin{figure}[tbp]
\vspace{-5mm}
    \centering
    \subfloat[]{%
       \includegraphics[width=\figwidthnormal]{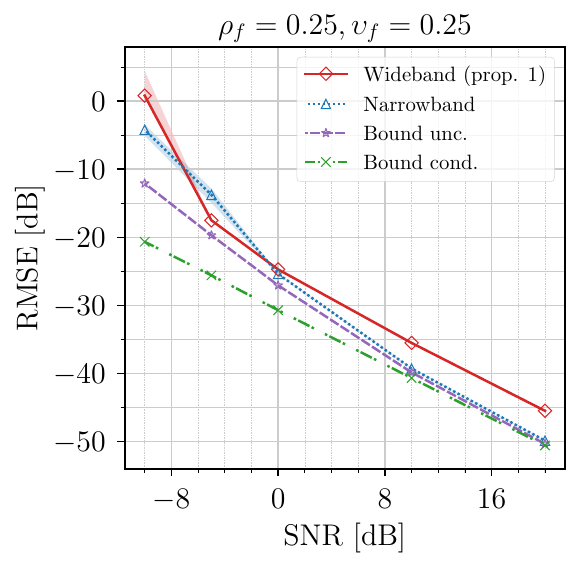}
       \label{fig:varcor_snr_lowrmse}}
    \hfill
    \subfloat[]{%
        \includegraphics[width=\figwidthnormal]{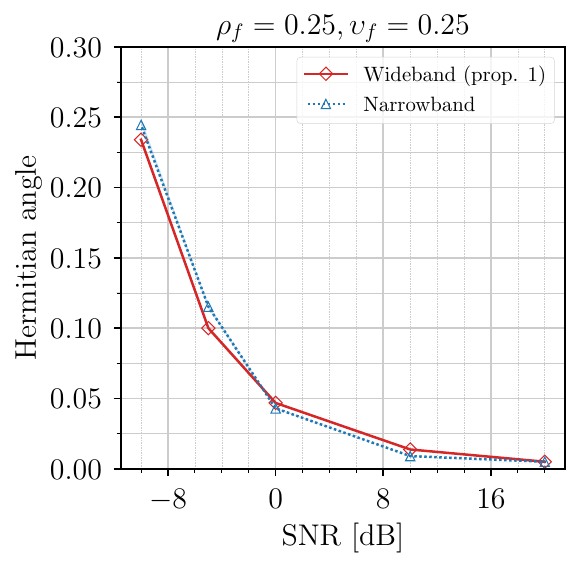}\label{fig:varcor_snr_lowha}}
    \hfill
    \\
    \subfloat[]{%
        \includegraphics[width=\figwidthnormal]{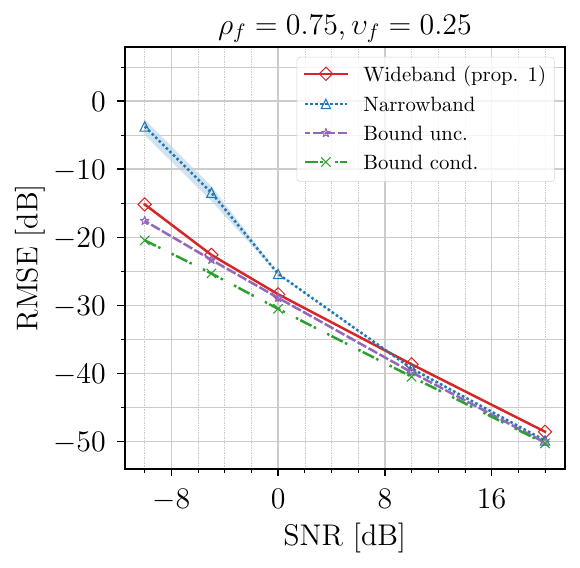}\label{fig:varcor_snr_highrmse}}
    \hfill
    \subfloat[]{%
       \includegraphics[width=\figwidthnormal]{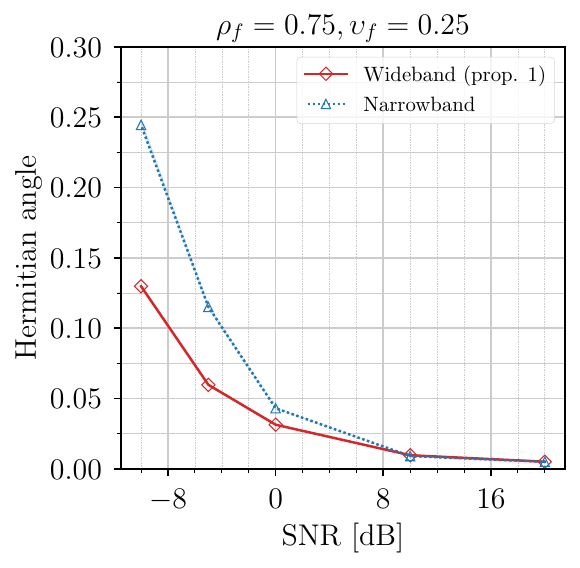}\label{fig:varcor_snr_highha}}
    \hfill
  \caption{Algorithm performance for non-uniform target and noise powers under varying SNR, with different levels of target correlation $\rho_f$. The top row ((a) and (b)) represents less correlated target ($\rho_f=0.25$), while the bottom row ((c) and (d)) shows highly correlated target ($\rho_f=0.75$). The left column corresponds to RMSE, and the right column shows the Hermitian angle.}
  \label{fig:varcor_snr}
\end{figure}
\blue{\subsection{Correlation coefficients of measured data}\label{sec:corr_analysis}\noindent
Before testing the RTF estimation algorithms on real data, it is useful to examine the correlation coefficients of measured adult speech and white Gaussian noise, and relate them to the simulated coefficients from \Cref{sec::equicorrelated,sec:varcorrelated}. 
To analyze the distribution of the spectral correlation coefficients, we first select a segment of length $T$ from the time-domain signal of interest.
After transforming this segment to the STFT domain, we estimate its spectral covariance matrix and spectral correlation coefficients, as detailed below.

The measured speech signal consists of anechoic \blue{male and female} speech recordings from the Harvard Word List\footnote{``Speech Intelligibility CD" from Neil Thompson Shade.}, sampled at $f_s = \SI{16}{\kilo\hertz}$.
The recordings last approximately 5 minutes.
For each of the $50$ Monte Carlo iterations, we randomly select a segment of length $T=\SI{0.35}{\second}$ from either of the two recordings.
Silent segments are discarded.
The STFT analysis is performed with window length $K_2 = 1024$, corresponding to $K_{+} = (K_2 / 2) + 1 = 513$ positive frequencies.
We use a square-root Hann window function and a \SI{75}{\percent} overlap between frames, such that the block-shift equals $R=256$ samples.
Therefore, the number of STFT frames available for estimating the spectral covariance matrices is $L \approx (T f_s - K_2 + R) / R \approx 20$.
The number of frames $L$ is thus small compared to the number of positive frequency bins $K_+$, complicating the estimation of the spectral covariance matrices.
To focus the analysis on relevant frequency bands, we only consider frequency components between \SIrange{0.08}{2.0}{\kilo\hertz}, reducing the number of frequency bins from $K_+=513$ to $K \approx 124$.
The remaining bins are ignored in the analysis.

The spectral covariance matrix $\hat{\v{R}} \in \mathbb{C}^{K \times K}$ is estimated from phase-adjusted STFT data, as described in \Cref{eq:phase_corr_cov}, to account for the phase shifts caused by overlapping frames. 
The magnitudes of the correlation coefficients for the off-diagonal elements are then obtained as:
\begin{align}
    \hat{\rho}_{f, k_1 k_2} = \Big|[\hat{\v{R}}]_{k_1 k_2} / \sqrt{\hat{\sigma}^2_{k_1 k_1} \hat{\sigma}^2_{k_2 k_2}}\Big|,
\end{align}
where $k_1, k_2 = 1,\ldots,K$, $k_1\neq k_2$, and $\hat{\sigma}^2_{k_1 k_1} = [\hat{\v{R}}]_{k_1 k_1}$.
Finally, we count the number of bins within each interval of the histogram and average the percentages across Monte Carlo realizations (\Cref{fig:hist}). 
As expected, the speech data (\Cref{fig:hist_speech}) exhibits significantly higher spectral correlation than the white noise data (\Cref{fig:hist_noise}). 
Approximately $20\%$ of the speech data shows correlations above $0.6$, while this value is nearly zero for the white noise data.
Surprisingly, the white noise data displays spectral correlations exceeding $0.2$ in over $40\%$ of the cases, which we hypothesize is due to spectral leakage caused by the finite-length windowing effect.
Given that the proposed RTF estimation method performs better for highly correlated target signals, the large number of low-correlation bins suggests that RTF estimation may improve by applying the proposed algorithm to the highly correlated bins only. 
The optimal design of such a method should be explored in future research.
}
\begin{figure}[tbp]
\vspace{-5mm}
    \centering
    \subfloat[]{%
        \includegraphics[width=0.48\linewidth]{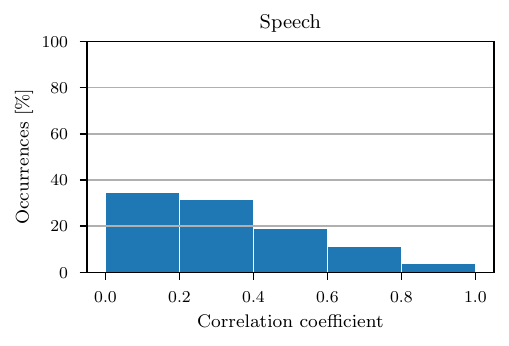}
        \label{fig:hist_speech}}
    \hfill
    \subfloat[]{%
       \includegraphics[width=0.48\linewidth]{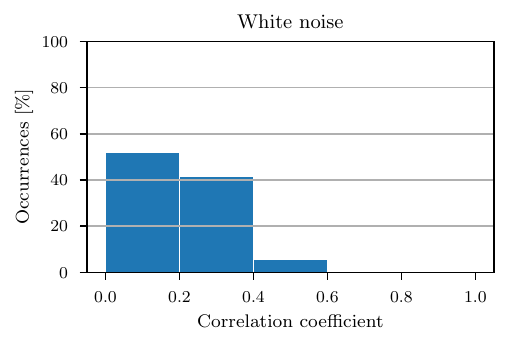}
       \label{fig:hist_noise}}
    \hfill
  \caption{Empirical distribution of spectral correlation coefficients for (a) male speech and (b) white noise signals. Speech has more  highly correlated bins than white noise, a stationary signal.}
  \label{fig:hist}
\end{figure}

\subsection{Real-data simulations}\label{sec:real_speech_exp}\noindent
The fourth set of experiments tests the RTF estimation algorithms on real speech data, maintaining the same settings as in \Cref{sec:corr_analysis}, except where noted.
\blue{A directional interferer is introduced by randomly sampling a real-world recording from the ESC-50 database \cite{piczak2015dataset}. 
Within the ESC-50 database, we sample from three selected categories  that contain approximately stationary sounds: engine noise, washing machine noise, and vacuum cleaner noise.
}
The default SNR for the interferer is set to \SI{0}{\decibel}.  
The target and the interfering signals are generated by convolution with the RIRs from the database in \cite{hadad_multichannel_2014}. 
The RIRs were measured with a linear microphone array with \SI{8}{} sensors, spaced \SI{8}{\cm} apart, in a room of size of \SI[parse-numbers=false]{6 \times 6 \times 2.4}{\meter}.
\blue{The average reverberation time of the room is $RT_{60} = \SI{0.61}{\second}$. All RIRs are cut after \SI{0.61}{\second} to reduce the length of the convolved signals while preserving most of the reverberation power.}
Unless otherwise specified, only the first $M=4$ microphones of the array are used.
The target and interfering sources are placed \SI{1}{\meter} away from the microphones, at angles of \SI{45}{\degree} and \SI{60}{\degree}, respectively.
Therefore, the target and interfering sources are spatially close to each other but exhibit different spectral properties.
\newcommand{\figwidth}{0.191\linewidth}
\begin{figure*}[t]
\vspace{-5mm}
    \centering
    \subfloat[]{%
        \includegraphics[width=\figwidth]{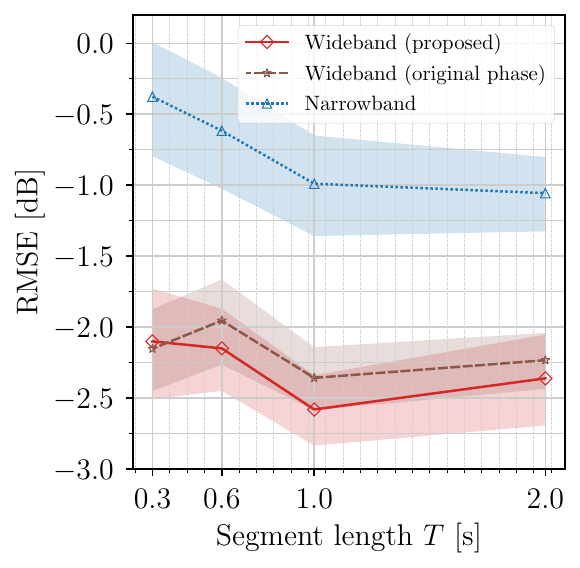}
        \label{fig:speech_timeframes_rmse}}
    \hfill%
    \subfloat[]{%
        \includegraphics[width=\figwidth]{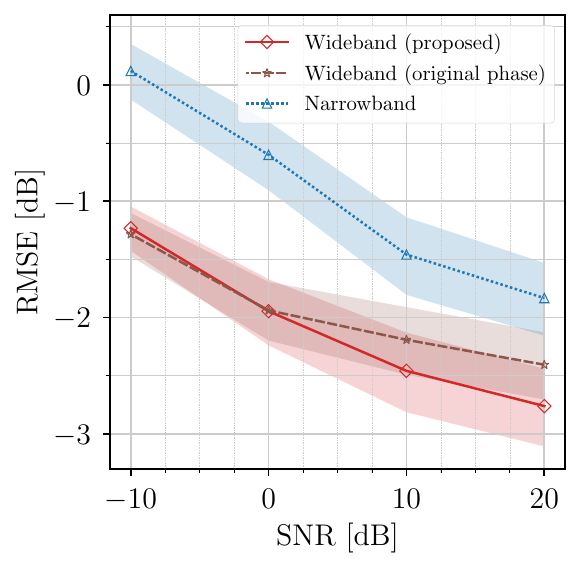}
        \label{fig:speech_snr_rmse}}
    \hfill%
    \subfloat[]{%
        \includegraphics[width=\figwidth]{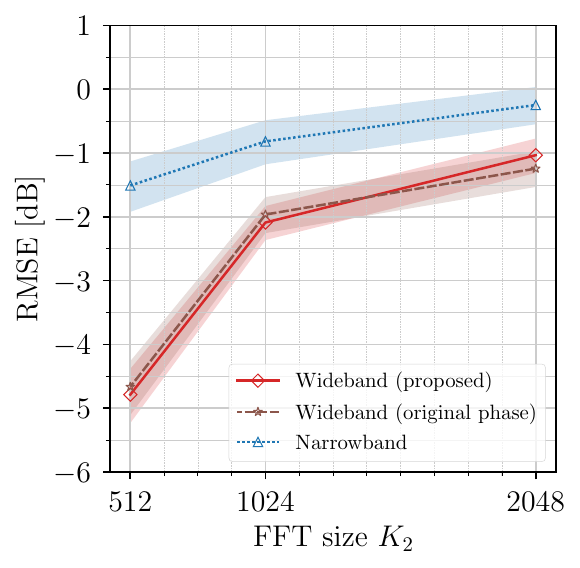}
        \label{fig:speech_nstft_rmse}}
    \hfill%
    \subfloat[]{%
        \includegraphics[width=\figwidth]{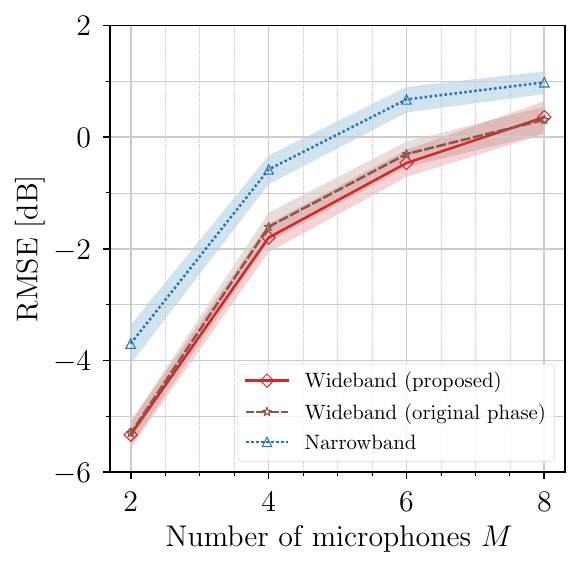}
        \label{fig:speech_nummics_rmse}}
    \hfill%
    \subfloat[]{%
    \includegraphics[width=\figwidth]{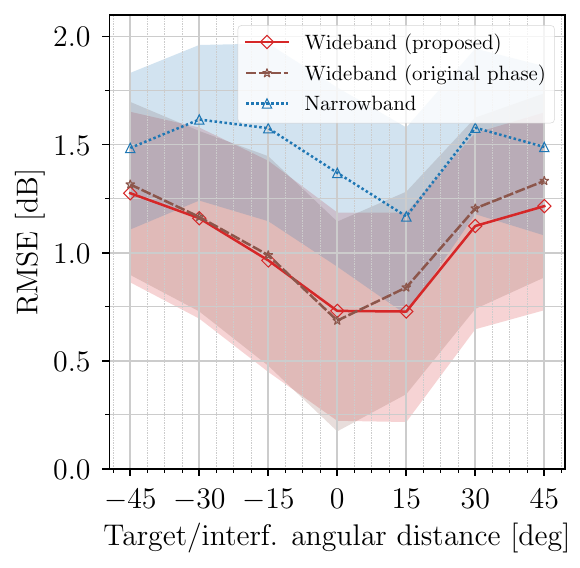}
    \label{fig:speech_noiseangle_rmse}}
    \hfill%
    \\
    \subfloat[]{%
       \includegraphics[width=\figwidth]{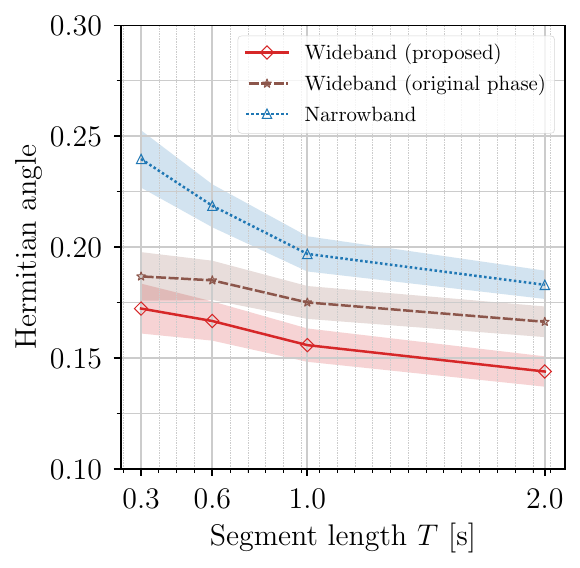}
       \label{fig:speech_timeframes_ha}}
   \hfill%
    \subfloat[]{%
       \includegraphics[width=\figwidth]{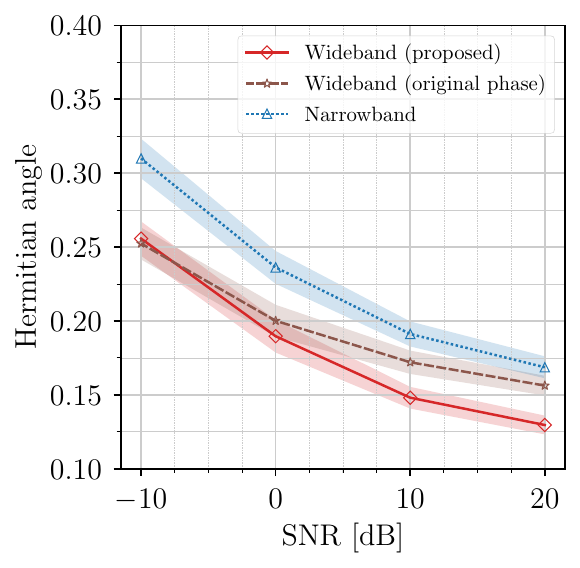}
       \label{fig:speech_snr_ha}}
    \hfill%
    \subfloat[]{%
       \includegraphics[width=\figwidth]{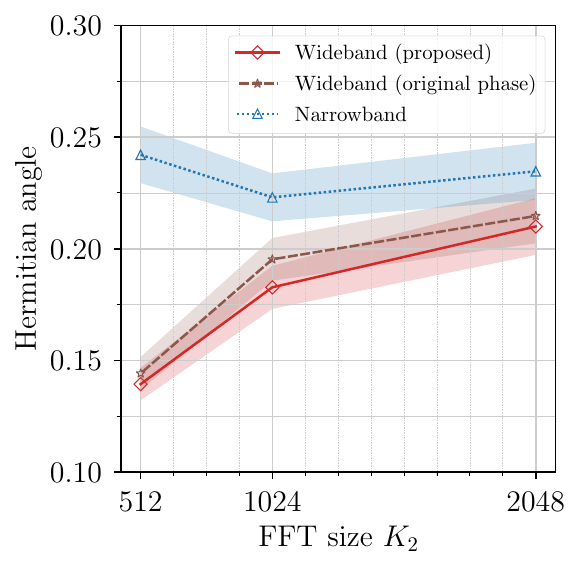}
       \label{fig:speech_nstft_ha}}
    \hfill%
     \subfloat[]{%
       \includegraphics[width=\figwidth]{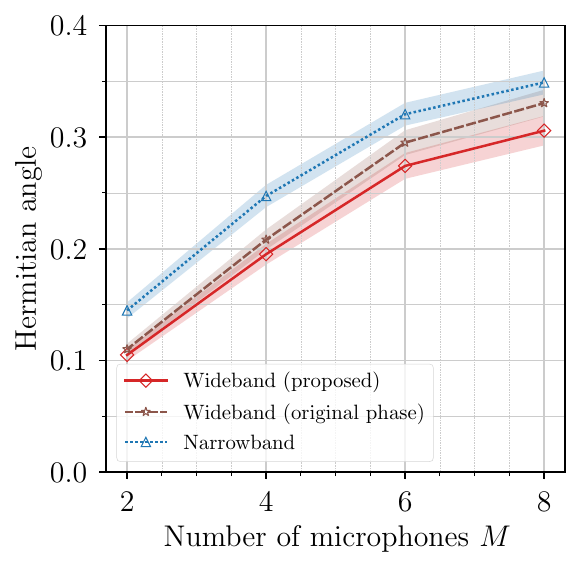}
       \label{fig:speech_nummics_ha}}
    \hfill%
    \subfloat[]{%
       \includegraphics[width=\figwidth]{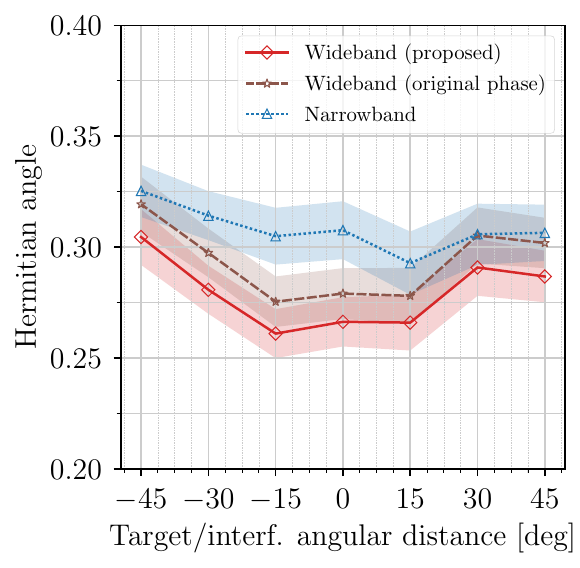}
       \label{fig:speech_noiseangle_ha}}
    \hfill%
  \caption{Performance of the algorithms for real speech, under (a-f) varying segment length $T$, (b-g) varying SNR, (c-h) varying FFT size, (d-i) varying number of microphones, and (e-j) varying angular distance between target and interferer. The error metrics are the RMSE (top) and the Hermitian angle (bottom).}
  \label{fig:exp_speech}
  \vspace{-5mm}
\end{figure*}
The noise covariance matrix $\hat{\v{R}}_v$ is estimated from a separate realization of the noise-only signal whose total duration is $T_n = \SI{2}{\second}$.
A distinct noise realization is used for each Monte Carlo iteration.
The covariance matrices $\hat{\v{R}}_x$ and $\hat{\v{R}}_v$ are estimated from phase-adjusted STFT data, as described in \Cref{eq:phase_corr_cov}, to account for the phase shifts caused by overlapping frames.
To gauge the improvements brought by the phase-corrected covariance estimator, we also depict the accuracy of the wideband SVD-direct algorithm when utilizing the sample covariance estimator of \Cref{eq:sample_cov_matrix}. 
Errors based on the sample covariance estimate with the original phase values are indicated by appending ``original phase" to the RTF estimator name.
When calculating the errors on the RTF estimates, we only retain the frequency bands for which the average power of the target signal at the microphones is no more than \SI{35}{\decibel} lower than that of the loudest frequency band.
Notice that the experiments include all bands in the specified frequency and SNR region. 
That means even those bands without correlation across frequency are included.

The wideband and narrowband algorithms are assessed based on the RMSE and the Hermitian angle metrics across different conditions.
\blue{We analyze the algorithms for different variations in the segment length $T$ (\Cref{fig:speech_timeframes_rmse,fig:speech_timeframes_ha}), which determines the number of frames $L$ available for estimating the covariance matrices, 
the SNR (\Cref{fig:speech_snr_rmse,fig:speech_snr_ha}),
the FFT size $K_2$ (\Cref{fig:speech_nstft_rmse,fig:speech_nstft_ha}), hence $K$,
the number of microphones $M$ (\Cref{fig:speech_nummics_rmse,fig:speech_nummics_ha}),
and the directional interferer angular position (\Cref{fig:speech_noiseangle_rmse,fig:speech_noiseangle_ha}).
}
The proposed phase-adjusted wideband algorithm outperforms the narrowband benchmark in all experiments of \Cref{fig:exp_speech}.
The performance gap between wideband and narrowband algorithms remains largely unchanged under varying conditions.
Notice that phase correction would not affect the performance of the narrowband algorithm, CW, which does not rely on inter-frequency correlations.
On the other hand, the wideband algorithm SVD-direct benefits significantly from incorporating phase-adjusted covariance matrices, especially in scenarios with a higher number of available time frames $L$ or microphones $M$, or when the SNR is high. 
The impact of phase correction appears to diminish under conditions where the covariance estimates are compromised due to a low number of time frames or a low SNR.
\blue{
In \Cref{fig:speech_nummics_rmse,fig:speech_nummics_ha}, we observe a decline in the performance of all algorithms as the number of microphones increases. 
This is likely because, as $M$ increases linearly, the number of elements in the covariance matrices ($M\times M$) increases quadratically. 
Consequently, with the data length remaining constant, the quality of the estimated covariance matrices worsens.
\Cref{fig:speech_noiseangle_rmse,fig:speech_noiseangle_ha} illustrate the RTF estimation errors when the target is positioned at $\SI{45}{\degree}$ and the interferer is placed at various angles within the \SIrange{0}{90}{\degree} range.
Although the wideband method outperforms the narrowband approach, the performance gap is greatest when the target and interferer are separated by a narrow angle, diminishing as the angular distance increases.
Exploiting spectral correlations proves to be most effective when there is significant overlap in the spatial correlations of the target and interferer.
}
\blue{
\subsection{Beamforming}\label{sec:beamforming_bf}\noindent
Knowledge of the RTF of a target speaker allows us to virtually steer a beamformer towards them and enhance the quality and the intelligibility of speech.
In this section, we evaluate the performance of an MVDR beamformer that uses the estimated RTFs to enhance a target signal, to get an impression of the performance improvement by using the proposed RTF estimator.
The output of the beamformer is an estimate of the target signal and its early reflections at the reference microphone, given by:
\begin{align}\label{eq:mvdr_out}
    \hat{d}_{1,k}(l) = \v{w}_k^H \v{x}_k(l),\quad l=1,\ldots,L~\text{and}~k=1,\ldots,K,
\end{align}
where $\v{w}_k \in \mathbb{C}^{M}$ are the beamforming weights, given by
\begin{align}\label{eq:mvdr}
    \v{w}_k = \frac{\hat{\v{R}}_n^{-1}(k) \hat{\v{a}}_k}{\hat{\v{a}}_k^H \hat{\v{R}}_n^{-1}(k) \hat{\v{a}}_k},\quad k=1,\ldots,K.
\end{align}
In \Cref{fig:bf}, we compare the output of the MVDR beamformer at various SNRs, using four different RTF estimates $\hat{\v{a}}_k$: SVD-direct (wideband), CW (narrowband), the true RTF $\v{a}_k$, and the unprocessed one at the reference microphone, \ie $\hat{\v{a}}_k = [1, 0, \ldots, 0]^T$. 
The output $\hat{d}_{1,k}(l)$ is evaluated using the short-time intelligibility index (STOI), the log-likelihood-ratio (llr) spectral distance, and the frequency-weighted segmental SNR (fwSNRseg) \cite{taal_algorithm_2011,loizou_2013}.
The same settings as in \Cref{sec:corr_analysis,sec:real_speech_exp} are used, except for the segment length $T$, which is extended to $T=\SI{1}{\second}$.
As expected, the MVDR beamformer performs best when using the true RTF, while using the noisy signal at the reference microphone results in the worst scores across all metrics.
The proposed algorithm generally outperforms CW according to all metrics in most conditions, except at very low SNRs.
}
\def\figbf{0.45\linewidth}
\begin{figure}[tbph]
\vspace{-3mm}
    \centering
    \subfloat[]{%
        \includegraphics[width=\figbf]{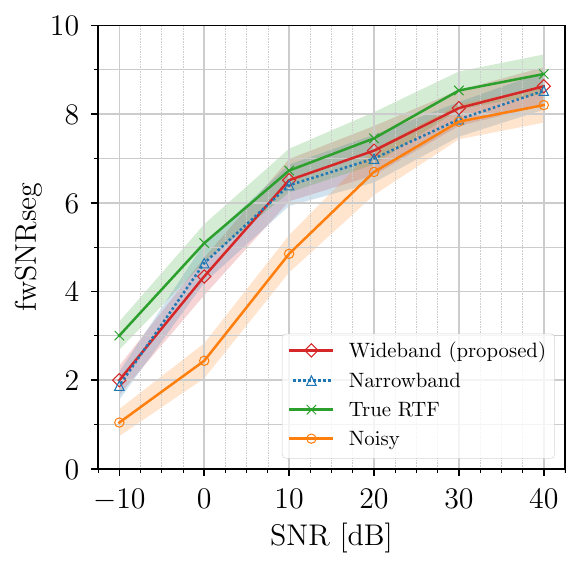} 
        \label{fig:bf_fwsnr}}
    \hfill%
    \subfloat[]{%
       \includegraphics[width=\figbf]{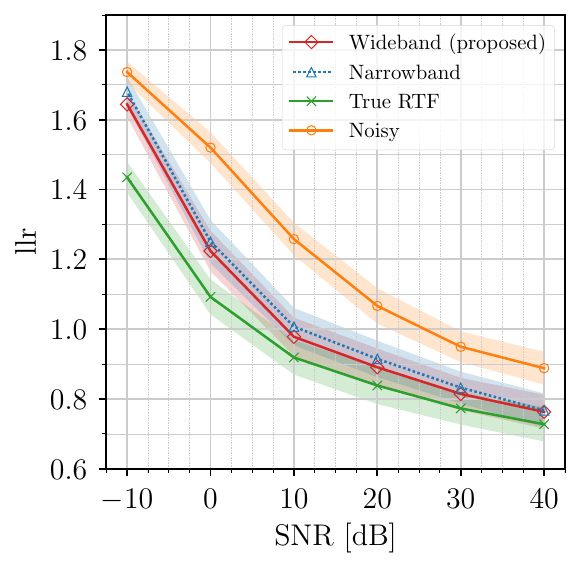}
       \label{fig:bf_llr}}
    \hfill%
    \\
    \centering
    \subfloat[]{%
        \includegraphics[width=\figbf]{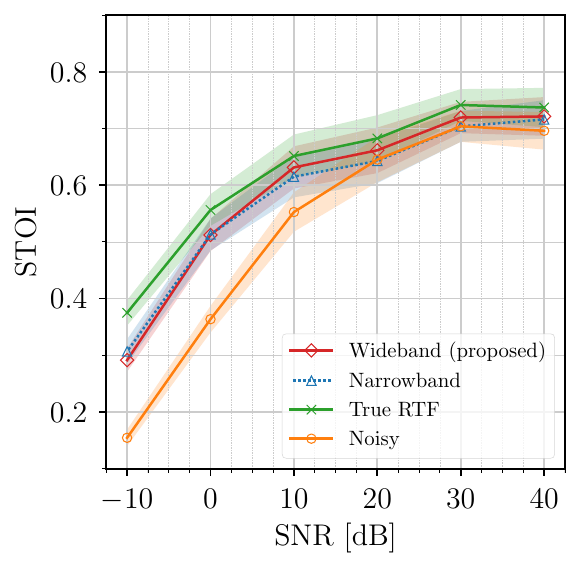}
        \label{fig:bf_stoi}}
    \hfill
  \caption{Evaluation of MVDR beamformer outputs using different RTF estimates, under varying SNR. Different subfigures correspond to different metrics: fwSNRseg (a), LLR (b), and STOI (c).}
  \label{fig:bf}
\end{figure}
\blue{
    \subsection{Computational complexity}\label{sec:complexity}\noindent
    Let us now compare the computational complexity of SVD-direct with the benchmark algorithm, CW.
    The cost of estimating the noisy covariance matrix is neglected as it is not considered part of the algorithms.
    While CW applies $K$ generalized eigendecompositions on spatial $M \times M$ covariance matrices, SVD-direct requires an initial generalized eigendecomposition on two matrices of size $KM \times KM$, followed by $K$ SVDs on matrices of size $M \times KM$. 
    Consequently, SVD-direct is slower than CW, primarily because the computational cost of eigenvalue decomposition increases cubically with the matrix size.
    Keeping the same settings as in the real-data experiments of \Cref{sec:real_speech_exp}, our measurements reveal that CW is approximately 200 times faster than SVD-direct using our non-optimized Python implementation on a MacBook Pro 16 inches (M1 Max chip).
}
\section{Additional discussion}\label{sec:discussion}\noindent
Our experiments yielded valuable insights into the performance of the narrowband CW and the wideband SVD-direct methods for RTF estimation, comparing them with the conditional and the unconditional performance bounds. Let us now examine the key findings.

Our investigation reveals a consistent trend favoring the wideband approach in scenarios with higher target correlation. Unlike the narrowband method, which remains unaffected by varying noise spectral correlation, the wideband approach also demonstrates occasional performance improvements when dealing with highly correlated noises.

The CRB analysis reveals a fascinating insight: when the noise spectral correlation $\upsilon_f$ is high, significant potential exists for further improvements in wideband channel estimation.
While this discovery pertains to channel estimation, it points to potential applications across various parameter estimation tasks.
The finding also provides a theoretical foundation for the observed empirical evidence in certain parametric and machine-learning approaches. Notably, some DNNs that operate on the entire time-frequency representation outperform narrowband alternatives in various speech enhancement tasks \cite{benesty_bifrequency_2012, tesch_insights_2023}.
It also offers a partial explanation for the intelligibility gains experienced by humans when detecting speech affected by harmonic noises \cite{gockel_asymmetry_2002}.

The correlation analysis of \Cref{sec:corr_analysis}, together with the real-speech experiments in \Cref{sec:real_speech_exp,sec:beamforming_bf}, confirm that natural speech possesses spectral correlations that can be exploited by the SVD-direct algorithm, leading to improved RTF estimation and beamforming performance in the scenarios under analysis.
However, it is worth noting that the wideband algorithm may not surpass narrowband algorithms in certain settings, such as when dealing with highly non-stationary noise sources.
This limitation arises from the inherent challenge of estimating large spectral-spatial covariance matrices from a limited number of frames.
For instance, typical speech has about 10-20 different sounds per second \cite[Chapter 15.3]{molisch_wireless_2011}.
This dynamic nature makes the estimation of spectral patterns more demanding compared to spatial patterns, which depend on the positions of the speaker and the listener.
Improved spectral correlation estimates may result by modeling the signals under analysis as realizations of cyclostationary processes, a particular class of spectrally correlated processes \cite{gardner_cyclostationarity_2006}.

\section{Conclusion}\label{sec:conclusions}\noindent
The uncorrelation of frequency components of a signal is a ubiquitous assumption that is often not verified in practice due to STFT processing and the non-stationary nature of signals.
In this paper, we investigated the role of spectral correlations in spatial processing and proposed a new subspace-based algorithm for the channel estimation task.
Indeed, accurate knowledge of the acoustic transfer functions between target speakers and microphones is crucial for spatial filtering in applications like MVDR beamforming.

Extensive numerical experiments demonstrated the superior performance of our wideband approach over the maximum-likelihood narrowband benchmark, yielding gains of more than \SI{10}{\decibel} RMSE in scenarios involving spectral correlations and low SNR. 
The proposed SVD-direct algorithm also exhibited competitive performance with real reverberant speech data contaminated by directional interferers and spatially uncorrelated noise. 
These achievements are obtained without compromising conceptual interpretability, as the channel estimate can be computed in closed form with just a few lines of code.

Furthermore, we derived CRBs for wideband channel estimation, revealing the potential for substantial accuracy improvements when noise, and to a lesser extent, the target exhibits high-frequency correlation.
This study serves as a starting point for understanding the impact of spectral correlations on parameter estimation for array processing.
Future endeavors will focus on refining the estimation of spectral-spatial covariance matrices, conducting more comprehensive tests on real-world measurements, and analyzing the influence of spectral correlation in speech enhancement and acoustic source separation.

\FloatBarrier
\appendices
\crefalias{section}{appendix}
\section{Proof of conditional CRB [\Cref{eq::crb_wideband_conditional}]}\label{app:proof_crb_cond}\noindent
\begin{proof}
To obtain the Cramér--Rao bound for the unknown parameters $\v{\theta}$, we first calculate the derivatives of the log-likelihood function with respect to the unknown parameters to form the Fisher information matrix.
Notice that the log-likelihood is a real-valued function of a complex variable $\v{\theta}$.
Thus, by evaluating the gradients using Wirtinger derivatives \cite{brandwood_complex_1983}, we can make use of the following properties.
\begin{lemma} \label{lemma:wirtinger}
    Let $f : \mathbb{C}^{p} \times \mathbb{C}^{p} \times \mathbb{C}^{q} \times \mathbb{C}^{q}  \mapsto \mathbb{R}$ be a real scalar function of four complex variables $\v{w}, \v{w}^* \in \mathbb{C}^{p}$ and $\v{z}, \v{z}^* \in \mathbb{C}^{q}$. Then
    \begin{enumerateproof}
        \item $\nabla_{\v{z}} f = (\grad_{{\v{z}}^*}f)^*$.
        \item $\nabla_{\v{z}} \nabla_{\v{w}}^H f = (\grad_{{\v{z}}^*} \grad_{{\v{w}}^*}^H f)^*$.
    \end{enumerateproof}
\end{lemma}
\begin{proof} \phantom{This text will be invisible}
\begin{enumerateproof}
        \item Follows from the fact the gradient operators are complex conjugates while $f$ is real.
        \item
        $\begin{aligned}[t]
            \nabla_{\v{z}} \nabla_{\v{w}}^H f &= \nabla_{\v{z}} \nabla_{{\v{w}}^*}^T f =
            (\nabla_{{\v{w}}^*} \nabla_{\v{z}}^T)^T f \\
            &= (\nabla_{{\v{w}}^*} \nabla_{{\v{z}}^*}^H)^T f = (\grad_{{\v{z}}^*} \grad_{{\v{w}}^*}^H f)^*.
        \end{aligned}$
    \end{enumerateproof}
    \vspace{-4mm}
\end{proof}
For notational convenience, we define $\mathcal{L}(\v{\theta}) = \ln{p(\v{X}; \v{\theta})}$.
We will begin by evaluating the bottom-right quadrant of the Fisher information matrix (\Cref{eq::fim}), defined as $-\E{\grad_{\v{a}^*}\grad^H_{\v{a^*}}\mathcal{L}(\v{\theta})}$.
Expanding the matrix product, the partial derivative of the log-likelihood $\mathcal{L}(\v{\theta})$ in \Cref{eq::likelihood_cond} with respect to $a_k^*$ is given by
\begin{align}
    \grad_{a_k^*} \mathcal{L}(\v{\theta}) &=
    -\grad_{a_k^*} \left(\sum_{l=1}^{L}(\v{x}(l) - \v{A}\v{s}(l))^H \v{R}_v^{-1} \v{v}(l)\right) \\
    &= \grad_{a_k^*} \left(\sum_{l=1}^{L} \sum_{j=1}^{KM} s_j^*(l) \v{a}_j^H \v{R}_v^{-1} \v{v}(l) \right) \\
    &= \sum_{l=1}^{L} s_k^*(l) \v{e}_k^T \v{R}_v^{-1} \v{v}(l).
    \end{align}
The second order derivative evaluates to
\begin{align}
    \grad_ {a_k^*}\grad_{a_m}\mathcal{L}(\v{\theta})
    &= -\sum_{l=1}^{L}s^*_k(l) \v{e}_k^T \v{R}_v^{-1} \v{e}_ms_m(l).
\end{align}
This leads to
\begin{align}
     -\E{\grad_{\v{a}^*}\grad^H_{\v{a^*}}\mathcal{L}(\v{\theta})} &= \sum_{l=1}^{L} \v{S}(l)^H \v{R}_v^{-1} \v{S}(l),
\end{align}
where we defined 
$\v{S}(l) = \diagp{\v{s}(l)}$.
With this, the Fisher information matrix is given by (cf. \Cref{lemma:wirtinger}) $\v{I}_{\v{\theta}} = \blkdiag(\v{B}^*, \v{B})$, where $\blkdiag(\cdot)$ is the operator that constructs a block diagonal matrix from the given matrices, and 
$\v{B} = \sum_{l=1}^{L} \v{S}(l)^H \v{R}_v^{-1} \v{S}(l)$.
The block-diagonal matrix $ \v{I}_{\v{\theta}}$ can be inverted block-wise, leading to 
\begin{align}\label{eq:inv_fish_cond_crb}
\v{I}_{\v{\theta}}^{-1} = \blkdiag((\v{B}^*)^{-1}, \v{B}^{-1}).
\end{align}
The variance of unbiased estimators of the ATF is, therefore, bounded by
$
    \text{var}(\hat{a}_i) \geq [(B^*)^{-1}]_{ii}, \quad i=1,\dots,M.
$

Transfer functions can be estimated in relation to a reference sensor $r$ with a function $\v{g}(\cdot)$ defined in \Cref{eq:g_def}.
Choosing $r=1$ as a reference sensor, the Jacobian matrix can be written as
\begin{align}\label{eq::dg_dtheta_wideband}
    \grad_{\v{\theta}}\v{g}
    &=
    \begin{bmatrix}
    \grad_{\v{a}}\v{g} & \grad_{\v{a}^*}\v{g}
    \end{bmatrix}
    =
    \begin{bmatrix}
    \grad_{\v{a}}\v{g} & \v{0}_{KM\times KM}
    \end{bmatrix},
\end{align}
where the right block of the gradient, $\grad_{\v{a}^*}\v{g}$, is null because $\v{g}(\cdot)$ does not depend on $\v{a}^*$.
We can further partition the left block of the gradient in $K$ ``fat" matrices $\grad_{\v{a}}\v{g} = [\grad_{\v{a}_1}^T\v{g}, \ldots, \grad_{\v{a}_K}^T\v{g}]^T$, where  $\grad_{\v{a}_k}\v{g} \in \mathbb{C}^{M\times KM},\ k=1,\dots,K$.
Individual blocks $\grad_{\v{a}_k}\v{g}$ can be written as
\begin{align}\label{eq:dg_dtheta_wideband_singleblk}
\scalemath{0.8}{
    \left[\begin{array}{@{}c|ccccc|c@{}}
    \multirow{5}{*}{$\v{0}_{M \times (k-1)M}$}& 0 & 0 & 0 & & 0 & \multirow{5}{*}{$\v{0}_{M \times (K - k)M}$}\\
    &-a_{k2}a_{k1}^{-2} & a_{k1}^{-1} & 0 & \dots & 0 &\\
    &-a_{k3}a_{k1}^{-2} & 0 & a_{k1}^{-1} & \dots & 0 &\\
    &\vdots &&& \ddots & \vdots & \\
    &-a_{kM}a_{k1}^{-2} & 0 & \dots & 0 & a_{k1}^{-1} &
    \end{array}\right],
    }
\end{align}
so that $\grad_{\v{a}}\v{g}$ shows a block-diagonal structure.
With this, we have from  \Cref{eq:inv_fish_cond_crb} and \Cref{eq:crb_det_function}
\begin{align}\label{eq:conditional_crb_final}
    \v{R}_{\hat{\v{\phi}}} \succeq (\grad_{\v{\theta}}\v{g})
    \v{I}^{-1}_{\v{\theta}}
    (\grad^H_{\v{\theta}}\v{g}) =
    (\grad_{\v{a}}\v{g})
    (\v{B}^*)^{-1}
    (\grad^H_{\v{a}}\v{g}).
\end{align}
The CRB corresponds to the diagonal elements of the matrix at the right-hand side of \Cref{eq:conditional_crb_final}, as stated in \Cref{eq::crb_wideband_conditional}.
\end{proof}
\section{Proof of unconditional CRB [\Cref{eq::crb_wideband_unconditional}]}\label{app:proof_crb_uncond}\noindent
We first list some derivative rules (\eg, \cite{kay_fundamentals_1993}) for a generic square matrix $\v{X}(\v{\theta})$, where the values of $\v{X}$ depend on $\v{\theta}$.
\begin{align}
    \grad_{\theta_i} \ln(\v{X}) &= \trace({\v{X}^{-1} \grad_{\theta_i} \v{X}}), \label{eq::derivative_log} \\
    \grad_{\theta_i} \trace(\v{X}) &= \trace(\grad_{\theta_i} \v{X}), \label{eq::derivative_trace} \\
    \grad_{\theta_i} \v{X}^{-1} &= -\v{X}^{-1} (\grad_{\theta_i} \v{X}) \v{X}^{-1}. \label{eq::derivative_inverse}
\end{align}
The computation of the unconditional CRB requires the calculation of first- and second-order derivatives of the log-likelihood in \Cref{eq::likelihood_uncond}, reproduced here for easy reference:
\begin{equation}\label{eq::likelihood_uncond_repeated}
\mathcal{L}(\v{\theta}) =
    -L\ln{|\pi\v{R}_x|} - L \trace{(\hat{\v{R}}_x\v{R}_x^{-1})}.
\end{equation}
\begin{proof}
The partial derivative of the log-likelihood $\mathcal{L}(\v{\theta})$ in \Cref{eq::likelihood_uncond} with respect to $a^*_k$ is given by
\begin{align}\label{eq::crb_uncond_first_derivative}
    \grad_{a^*_k}\mathcal{L}(\v{\theta}) = -L\trace{(\v{R}_x^{-1}\v{F}_k)} - L \trace{(\v{R}_x^{-1}\v{F}_k\v{R}_x^{-1}\hat{\v{R}}_x)},
\end{align}
where we introduced
\begin{align}\label{eq::f_k_definition}
    \v{F}_k = \grad_{a^*_k} \v{R}_x = \v{A}\v{R}_s \v{E}^{kk}
\end{align}
and $\v{E}^{ij}$ is zero everywhere and 1 at entry $ij$.
The first term on the right-hand side of \Cref{eq::crb_uncond_first_derivative} is obtained directly from \Cref{eq::derivative_log}.
The second term is obtained by using the derivative of the trace and of the matrix inverse as given by \Cref{eq::derivative_trace} and \Cref{eq::derivative_inverse}, respectively, along with the cyclic property of the trace operator.
It is important to note that the estimate $\hat{\v{R}}_x$ is independent of the parameter vector $\v{a}$.
The second-order partial derivative of $\mathcal{L}(\v{\theta})$ writes
\begin{multline}
\label{eq::crb_uncond_second_derivative_initial}
    \grad_{a_m} \grad_{a^*_k} \mathcal{L}(\v{\theta}) = \\-L\grad_{a_m}[\trace{(\v{R}_x^{-1}\v{F}_k)} + \trace{(\v{R}_x^{-1}\v{F}_k\v{R}_x^{-1}\hat{\v{R}}_x})].
\end{multline}
The derivatives of the two terms can be evaluated separately.
For the first term in \Cref{eq::crb_uncond_second_derivative_initial}, we have
\begin{multline}\label{eq::crb_uncond_second_derivative_first_term}
    \grad_{a_m}\trace{(\v{R}_x^{-1}\v{F}_k)} = \trace{(\grad_{a_m}\v{R}_x^{-1}\v{F}_k)} =\\
    = \trace{(-\v{R}_x^{-1}\v{G}_m \v{R}_x^{-1} \v{F}_k + \v{R}_x^{-1}\v{H}_{mk})},
\end{multline}
which is obtained by applying the product rule together with \Cref{eq::derivative_trace} and \Cref{eq::derivative_inverse}.
Here, $\v{G}_m$ and $\v{H}_{mk}$ are defined as
\begin{gather}
\v{G}_m = \grad_{a_m} \v{R}_x = \v{E}^{mm} \v{R}_s \v{A}^H,\label{eq::g_m_definition}\\
\v{H}_{mk} = \grad_{a_m} \v{F}_k = \grad_{a_m}\grad_{a^*_k} \v{R}_x = \v{E}^{mm} \v{R}_s \v{E}^{kk}.
\end{gather}
The second term in \Cref{eq::crb_uncond_second_derivative_initial} is given by
\begin{multline}\label{eq::crb_uncond_second_derivative_second_term}
    \grad_{a_m}\trace{(\v{R}_x^{-1}\v{F}_k\v{R}_x^{-1}\hat{\v{R}}_x}) =
    \trace{\grad_{a_m}(\v{R}_x^{-1}\hat{\v{R}}_x}\v{R}_x^{-1}\v{F}_k) =\\
    \trace{\left[
    \v{R}_x^{-1} \hat{\v{R}}_x \v{R}_x^{-1}
    (\v{G}_m \v{R}_x^{-1} \v{F}_k - \v{H}_{mk}
    + \v{F}_k \v{R}_x^{-1} \v{G}_m)
    \right]},
\end{multline}
which was again obtained by utilizing the product rule, \Cref{eq::derivative_trace}, \Cref{eq::derivative_inverse}, and rearranging the resulting terms.
By combining \Cref{eq::crb_uncond_second_derivative_first_term} and \Cref{eq::crb_uncond_second_derivative_second_term}, the negative expected second-order partial derivative follows as
\begin{multline}\label{eq:}
    -\E{\grad_{a_m} \grad_{a^*_k} \mathcal{L}(\v{\theta})} = \\
    L \expectedValue \Big\{
    \, \trace{\left[
    -\v{R}_x^{-1}(\v{G}_m \v{R}_x^{-1} \v{F}_k - \v{H}_{mk})
    \right]}\ + \\
    + \trace{\left[
    \v{R}_x^{-1} \hat{\v{R}}_x \v{R}_x^{-1}
    (\v{G}_m \v{R}_x^{-1} \v{F}_k - \v{H}_{mk}
    + \v{F}_k \v{R}_x^{-1} \v{G}_m)
    \right]}\Big\} \\
    = L \trace{\left(
            \v{R}_x^{-1} \v{F}_k \v{R}_x^{-1} \v{G}_m
            \right)}
    .
\end{multline}
To collect the expected values of the second-order partial derivative, we define a matrix $\v{C}_1$ such that $[\v{C}_1]_{mk} = -\E{\grad_{a_m}\grad_{a^*_k} \mathcal{L}(\v{\theta})}$.
The elements of the bottom left block of the Fisher information matrix can be similarly obtained as $[\v{C}_2]_{mk} = -\E{\grad_{a_m}\grad_{a_k} \mathcal{L}(\v{\theta})} = L \trace{\left(
\v{R}_x^{-1} \v{G}_k \v{R}_x^{-1} \v{G}_m
\right)}$,
and the elements $-\E{\grad_{a_m^*}\grad_{a_k^*} \mathcal{L}(\v{\theta})}$ of the top right block follow as $\v{C}_2^H$ from \Cref{lemma:wirtinger}.
The inverse of the Fisher information matrix for the unconditional case can then be represented by:
\begin{align}\label{eq:inv_fish_uncond_crb}
    \v{I}_{\v{\theta}}^{-1} =
    \begin{bmatrix}
    \v{C}_1^* & \v{C}_2^H \\
        \v{C}_2 & \v{C}_1
    \end{bmatrix}^{-1}
    =
    \begin{bmatrix}
    \v{C} & * \\
    * & *
    \end{bmatrix},
\end{align}
where $\v{C} \in \mathbb{C}^{KM\times KM}$ is obtained by selecting the first $KM$ rows and columns from $\v{I}_{\v{\theta}}^{-1}$.
To derive the bound for estimating the \textit{relative} transfer function, we employ the mapping $\v{g}(\cdot)$ as defined in \Cref{eq:g_def}. Using \Cref{eq:crb_det_function}, the inverse Fisher information matrix is left- and right-multiplied by $\grad_{\v{\theta}}\v{g}$, which is defined in \Cref{eq::dg_dtheta_wideband}, resulting in the final form of the bound given by \Cref{eq::crb_wideband_unconditional}.
\end{proof}

\bibliographystyle{IEEEtran}

\begin{IEEEbiography}[{\includegraphics
[width=1in,height=1.25in,clip,
keepaspectratio]{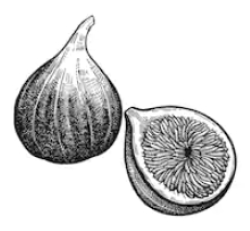}}]
{Giovanni Bologni} received his BSc in Electronics Engineering from the Polytechnic University of Turin, Italy, in 2017 and his MSc in Electrical Engineering from the Delft University of Technology, Netherlands, in 2020.
From 2020 to 2022, he worked at Sony R\&D Stuttgart, Germany, initially as an intern and later as a signal processing engineer.
He is currently pursuing a PhD in the Signal Processing Systems (SPS) group at Delft University of Technology. 
His research focuses on statistical signal processing for audio and speech, with an emphasis on multichannel speech enhancement for hearing aids.

\end{IEEEbiography}

\end{document}